\documentclass[a4paper]{article}
\usepackage[T1]{fontenc}

\usepackage{amssymb,amsmath,amsthm}
\usepackage{tikz}

\usepackage{mathtools}
\usepackage{thmtools} 
\usepackage{thm-restate}

\usepackage[hidelinks, colorlinks, allcolors=blue]{hyperref}
\usepackage{cleveref}

\usepackage{fullpage}

\usepackage[linesnumbered,lined,boxed,commentsnumbered,noend]{algorithm2e}

\title{Improved Space Bounds for Subset Sum}

\author{
Tatiana Belova
\thanks{St.~Petersburg Department of~Steklov Mathematical Institute, Russian Academy of~Sciences. Email: \texttt{yukikomodo@gmail.com}.}
\and
Nikolai Chukhin 
\thanks{Neapolis University Pafos and JetBrains Research. Email: \texttt{buyolitsez1951@gmail.com}}
\and
Alexander~S. Kulikov
\thanks{JetBrains Research. Email: \texttt{alexander.s.kulikov@gmail.com}.}
\and
Ivan Mihajlin
\thanks{JetBrains Research. Email: \texttt{ivmihajlin@gmail.com}.}
}
\date{}

\newcommand{\poly}{\operatorname{poly}}
\newcommand{\SAT}{\text{\textsf{SAT}}}
\newcommand{\AM}{\text{\textsf{AM}}}
\newcommand{\SUM}{\text{\textsf{SUM}}}
\newcommand{\problem}[1]{\text{\textsf{#1}}}

\newtheorem{openproblem}{Open Problem}[section]
\newtheorem{theorem}{Theorem}[section]
\newtheorem{lemma}{Lemma}[section]
\newtheorem{definition}{Definition}[section]
\newtheorem{corollary}{Corollary}[section]

\begin{document}

\sloppy

\maketitle

\begin{abstract}
More than 40~years ago, Schroeppel and Shamir presented 
    an~algorithm that solves the Subset Sum problem for 
    $n$~integers in~time $O^*(2^{0.5n})$ and space $O^*(2^{0.25n})$.
    The time upper bound remains unbeaten, 
    but the space upper bound has been improved 
    to~$O^*(2^{0.249999n})$
    in a~recent breakthrough paper
    by~Nederlof and W\k{e}grzycki (STOC 2021).
    Their algorithm is~a~clever combination 
    of~a~number of previously known techniques 
    with a~new reduction and a~new algorithm
    for the Orthogonal Vectors
    problem.

    In~this paper, we~give two new algorithms for Subset Sum.
    We~start by~presenting an~Arthur--Merlin algorithm:
    upon receiving the verifier's randomness,
    the prover sends an~$n/4$-bit long
    proof to~the verifier who checks~it
    in~(deterministic) time and space $O^*(2^{n/4})$.
    An interesting consequence of~this result
    is~the following fine-grained lower bound: assuming that $4$-SUM cannot be~solved
    in~time $O(n^{2-\varepsilon})$ for all $\varepsilon>0$, Circuit SAT cannot be~solved
    in~time $O(g2^{(1-\varepsilon)n})$, for all~$\varepsilon>0$ (where $n$~and~$g$ denote the
    number of~inputs and the number of~gates, respectively).

    Then, we~improve the space bound by~Nederlof and W\k{e}grzycki
    to~$O^*(2^{0.246n})$ and also simplify their algorithm
    and its analysis.  
    We~achieve this space bound by~further filtering sets of~subsets
    using a~random prime number.
    This allows~us to~reduce an~instance of~Subset Sum 
    to a~larger number of~instances of~smaller size.
\end{abstract}

 \clearpage
 \tableofcontents
 \clearpage

\section{Overview}
In~this paper, we~study the well-known \problem{Subset Sum} problem and its parameterized version, $k$-\SUM{}. 
In~\problem{Subset Sum}, given a~set
of~$n$~integers 
and a~target integer~$t$
the goal is~to~check 
whether there is a~subset of~them
that sum~up to~$t$. It~is common to~assume
that the absolute value of~all input integers is~at~most $2^{O(n)}$ (one can achieve this using hashing techniques).

In the~$k$-\SUM{} problem, the goal is~to~check whether
some~$k$ of~the~input integers sum to~zero. We~consider 
a~slightly different formulation of~$k$-\SUM{}:
given $k$~sequences $A_1, \dotsc, A_k$,
each containing $n$~integers of~absolute value at~most $n^{O(1)}$, 
check whether there exist indices $i_1, \dotsc, i_k$
such that\footnote{These two versions of $k$-\SUM{} are reducible 
	to~each other. One direction is~easy: take $k$~copies of~the array~$A$. For the other direction, one can use the color coding technique~\cite{DBLP:journals/jacm/AlonYZ95}:
	for every element~$A$, add it~to $A_i$ for random $i \in [k]$; then, if a~solution for the original problem exists,
it~survives in~the resulting instance with constant probability. This randomized reduction can also be~derandomized~\cite{DBLP:journals/jacm/AlonYZ95}.}
\(A_1[i_1]+\dotsb+A_k[i_k]=0\).
It~is common to~further assume that the bit-length of~all integers in~the $k$-\SUM{} problem is~at~most $k\log n + O(\log \log n)$. 
This can be~achieved by the standard fingerprinting technique.

\subsection{Known Results} 

\subsubsection{Time Complexity}

$2$-\SUM{} can be~solved in~time $\widetilde{O}(n)$ using binary search (the $\widetilde{O}(\cdot)$ notation hides polylogarithmic factors): sort~$A_1$; then,
for each $1 \le i_2 \le n$, check whether $A_1$~contains
an~element $-A_2[i_2]$. Another well known algorithm
is~based on~the two pointers method that proceeds as~follows.
Sort $A_1$ and~$A_2$
and let $i_1=1$ and $i_2=n$.
Then, keep repeating the following:
	if $i_1 > n$ or $i_2 < 1$, stop; 
	if $A_1[i_1]+A_2[i_2]=0$, return $(i_1, i_2)$;
	if $A_1[i_1]+A_2[i_2]>0$, decrement $i_2$ (as~$A_2[i_2]$ is~too large: its sum with the currently smallest element of~$A_1$ is~positive);
	if $A_1[i_1]+A_2[i_2]<0$, increment $i_1$ (as~$A_1[i_1]$ is~too small: its sum with the currently largest element of~$A_2$ is~negative).

The algorithm for $2$-\SUM{} 
allows $k$-\SUM{} to~be~solved
 in time $\widetilde{O}(n^{\lceil k/2 \rceil})$ 
via the following reduction. Given $k$~arrays $A_1, \dotsc, A_k$, split them into two halves: 
$A_1, \dotsc, A_{\lfloor k/2 \rfloor}$
and 
$A_{\lceil k/2 \rceil}, \dotsc, A_k$. 
Populate 
arrays $B_1$~and~$B_2$ with all sums of~elements from the 
first and second halves, respectively. Then, it~remains
to~solve \problem{$2$-SUM} for $B_1$~and~$B_2$.

{For $3$-\SUM{}, 
Chan~\cite{DBLP:journals/talg/Chan20} proved a~slightly better upper bound $O(n^2\log \log^{O(1)} n/\log^2 n)$}.
It~is a~major open problem whether $3$-\SUM{} can be~solved
in~time $O(n^{2-\varepsilon})$ and the $3$-\SUM{} hypothesis, stating that this is~impossible, 
is~a~popular hypothesis in~the field of~fine-grained complexity. It~is also currently unknown whether $k$-\SUM{}
can be~solved in~time $\widetilde{O}(n^{\lceil k/2 \rceil - \varepsilon})$ for any $k \ge 3$ and $\varepsilon>0$.

The $2$-\SUM{} algorithm described above is~also
at~the core of~the strongest known upper bound $O^*(2^{n/2})$ for \problem{Subset Sum} ($O^*(\cdot)$ hides polynomial factors) 
presented by~Horowitz and Sahni~\cite{DBLP:journals/jacm/HorowitzS74}
fifty years ago. 
To~get this running time, partition~$A$ into two halves
of~size~$n/2$; then, populate arrays $A_1$ and $A_2$ (of~size $2^{n/2}$) with sums of~all subsets of~the two halves, respectively; then, it~remains to~solve $2$-\SUM{}
for $A_1$~and~$A_2$. It~is an~important open problem 
to~solve \problem{Subset Sum} in $O^*(2^{(1/2-\varepsilon) n})$
for a~constant $\varepsilon>0$. The fastest known algorithm shows that \problem{Subset Sum} can be solved in time $2^{n/2} / \poly(n)$ \cite{chen2023subset}.

\subsubsection{Space Complexity}

\subsubsection{Space Complexity}

The algorithms that solve $k$-\SUM{} and \problem{Subset Sum} via a~reduction to~$2$-\SUM{} have high space complexity: 
for $k$-\SUM{}, it~is $O(n^{\lceil k/2 \rceil})$,
whereas for \problem{Subset Sum} it is $O(2^{n/2})$.
It~is natural to~ask whether one can lower the space complexity or~whether it~is possible to~trade off time and space. Lincoln et~al.~\cite{DBLP:conf/icalp/LincolnWWW16} gave a~positive answer for $k$-\SUM{}: one can solve
$3$-\SUM{} in~time $O(n^2)$ and space $O(\sqrt n)$
as~well~as $k$-\SUM{} (for $k \ge 4$) in~time $O(n^{k-2+2/k})$ and space~$O(\sqrt n)$.
For \problem{Subset Sum}, 
the well-known algorithm
by~Schroeppel and Shamir~\cite{DBLP:journals/siamcomp/SchroeppelS81}
solves~it
in~time $O^*(2^{n/2})$ and space $O(2^{n/4})$:
they reduce \problem{Subset Sum} to~$4$-\SUM{}
and note that $4$-\SUM{} can be~solved in~time $\widetilde{O}(n^2)$ and space $O(n)$
(since all pairwise sums 
of~two sorted sequences of~length~$n$ can be~enumerated in~time~$\widetilde{O}(n^2)$ and space~$O(n)$ using a~priority queue data structure).
Just recently, 
Nederlof and Wegrzycki improved the space complexity
to~$O(2^{0.249999n})$~\cite{DBLP:conf/stoc/NederlofW21}.

\subsubsection{Proof Complexity}
It~is easy to~certify a~yes-instance of~$k$-\SUM{} or~\problem{Subset Sum}: just write down a~solution; 
it~can be~checked in~deterministic time~$O(n)$.
Certifying no-instances is~trickier.
Carmosino et al.~\cite{DBLP:conf/innovations/CarmosinoGIMPS16} 
showed that, for $3$-\SUM{}, there are proofs of~size $\widetilde{O}(n^{1.5})$ that can be~deterministically
verified in~time $\widetilde{O}(n^{1.5})$.
By~allowing the verifier to~be probabilistic,
Akmal et al.~\cite{DBLP:conf/innovations/Akmal0JR022}
presented a~proof system for $3$-\SUM{} where the 
proof's size and the verification time is~$\widetilde{O}(n)$.
For $k$-\SUM{}, they give an~upper bound of $\widetilde{O}(n^{k/3})$. 
The corresponding proof system
is known as~a~Merlin--Arthur protocol: a~computationally unbounded prover (Merlin) prepares a~proof 
and sends it~to a~probabilistic verifier (Arthur)
who needs to~be able to~check the proof quickly
with small constant error probability. We~say that
a~problem can be~solved in Merlin--Arthur time $T(n)$,
if~there exists a~protocol where both the proof size and the verification time do~not exceed~$T(n)$. 

For \problem{Subset Sum},
Nederlof~\cite{DBLP:journals/ipl/Nederlof17} 
proved an~upper bound $O^*(2^{(1/2-\varepsilon)n})$, for some $\varepsilon>10^{-3}$, on~Merlin--Arthur time.
Akmal et al.~\cite{DBLP:conf/innovations/Akmal0JR022}
improved the bound to~$O^*(2^{n/3})$.

\subsection{New Results} 

Below, we~give an~overview of~the main results of~the paper.

\subsubsection{New Arthur--Merlin Algorithm}

Recall that in an~Arthur--Merlin protocol the randomness is~shared 
(also known as~public coins)
and the verifier is~deterministic:
upon receiving the verifier's randomness,
the prover prepares a~proof and sends~it
to~the verifier who checks~it deterministically
with small error probability (taken over public randomness).
Formally, we~say that a~language $L \subseteq \{0,1\}^*$
belongs to a~class $\AM[s(n), t(n)]$, if~there exists
an~Arthur--Merlin protocol such that, for any $x \in \{0,1\}^n$, the proof length is~at~most $s(n)$, the verification time is~at~most $t(n)$, and
for each $x \in L$, the verifier accepts with probability at~least~$2/3$, whereas for each $x \not \in L$, the verifier rejects with probability~$1$.\footnote{
	In~the standard definition of~AM protocols,
	there is either a~two-sided error probability or zero error probability of~acceptance. We~choose to~have a~zero error probability of~rejection as~it allows for
	straightforward transformation to~randomized and parallel algorithms.}
Such a~protocol implies that $L$~can be~solved 
by~a~randomized algorithm that has running time $t(n)$
and error probability at~least $2^{-s(n)}$. It~is
also easily parallelizable: 
using $k \le 2^{s(n)}$~processors, the problem can be~solved in time $\frac{2^{s(n)}t(n)}{k}$.

Our first main result is a~new Arthur--Merlin algorithm for \problem{Subset Sum}.

\begin{restatable}{theorem}{subsetsumtheorem}
	\label{thm:subsetsum}
	\(\problem{Subset Sum} \in \AM\left[ n/4, O^*(2^{n/4})\right].\)
\end{restatable}

As~discussed above, 
it can be~parallelized easily: 
to~solve the problem, one can
enumerate possible proofs in~parallel.
Also, by~enumerating 
all possible proofs, one recovers upper bounds
on~time and space for \problem{Subset Sum} proved 
by~Schroeppel and Shamir in~1979.
Interestingly, the resulting algorithm is~very simple
and does not need to~use the priority queue data structure
as in~the algorithm by~Schroeppel and Shamir.

As~it is~the case with the previously known algorithms
for \problem{Subset Sum}, our algorithm 
follows from 
an~algorithm for $4$-\SUM.

\begin{restatable}{theorem}{foursumtheorem}
	\label{thm:foursum}
	\(\text{$4$-\SUM{}} \in \AM\left[ \log_2 n, \widetilde{O}(n)\right].\)
\end{restatable}

Since $2k$-\SUM{} can be~reduced to~$4$-\SUM{}, this extends to~$2k$-\SUM{} as~follows.

\begin{restatable}{corollary}{twoksumcorollary}
	For every even integer~$k$,
	\label{thm:twoksum}
	\(\text{$2k$-\SUM{}} \in 
	\AM\left[ 
		\frac{k}{2}\log_2 n, 
		\widetilde{O}(n^{k/2})
	\right].\)
\end{restatable}

The main idea of~the proof of~Theorem~\ref{thm:foursum}
is~the following. To~find integers $i_1,i_2,i_3,i_4$ such that $A_1[i_1]+A_2[i_2]+A_3[i_3]+A_4[i_4]=0$, one generates a~random prime $p \le n$ and uses $(A_1[i_1]+A_2[i_2]) \bmod p$ as a~proof.
A~similar idea of~using a~random prime for filtering subsets
was used previously
by~Howgrave-Graham and Joux~\cite{howgrave2010new}.

\subsubsection{New Conditional Lower Bounds for Circuit SAT}

Another interesting consequence of~Theorem~\ref{thm:foursum}
is~a~fine-grained lower bound for the \problem{Circuit SAT} problem. 
\problem{Circuit SAT} is a~generalization 
of~\problem{SAT} where instead of~a~formula in~CNF,
one is given an~arbitrary Boolean circuit with $g$~binary gates and $n$~inputs (and the goal,
as~usual, is~to~check whether it~is satisfiable).
Being a~generalization of~\problem{SAT}, \problem{Circuit SAT} cannot be~solved in time $O^*(2^{(1-\varepsilon)n})$, for any constant $\varepsilon>0$, under the Strong Exponential Time Hypothesis (that states that for every $\varepsilon>0$, there exists~$k$ such that $k$-\SAT{} cannot be~solved
in time~$O(2^{(1-\varepsilon)n})$).
Even designing a~$2^{(1-\varepsilon)n}$ time algorithm 
solving \problem{Circuit SAT} for circuits with $g \le 8n$
gates is~challenging: as~shown
by~\cite{DBLP:journals/iandc/JahanjouMV18}, this would imply new circuit lower bounds. 

We~prove the following conditional lower bound for 
\problem{Circuit SAT}.

\begin{restatable}{theorem}{cktsattheorem}
	\label{thm:cktsat}
	If, for any $\varepsilon>0$, $4$-\SUM{} cannot 
	be~solved in~time $O(n^{2-\varepsilon})$,
	then, for any $\varepsilon>0$, \problem{Circuit SAT}
	cannot be~solved in~time $O(g2^{(1-\varepsilon)n})$.
\end{restatable}

This theorem however does not exclude a~$O(g^{O(1)}2^{(1-\varepsilon)n})$-time algorithm for
\problem{Circuit-SAT} in~the same manner. To~get such a~fine-grained lower bound, one needs to~push the algorithm from \Cref{thm:twoksum} further (i.e., to trade most of~the $n^k$ time upper bound for nondeterminism).

\begin{openproblem}
	Prove or~disprove:
	\(\text{$2k$-\SUM{}} \in 
	\AM\left[ 
	k\log_2 n, 
	n^{o(k)}
	\right].\)
\end{openproblem}

\begin{restatable}{corollary}{cktsattwokcorollary}
	\label{thm:cktsattwok}
	Assume that
	\(\text{$2k$-\SUM{}} \in 
	\AM\left[ 
	k\log_2 n, 
	n^{o(k)}
	\right].\)
	Then, if for any $\varepsilon>0$ and any $k \ge 1$, $2k$-\SUM{} cannot 
	be~solved in~time $O(n^{k-\varepsilon})$,
	then, for any $\varepsilon>0$, \problem{Circuit SAT}
	cannot be~solved in~time $O(g^{O(1)}2^{(1-\varepsilon)n})$.
\end{restatable}

\color{black}

\subsubsection{Improved Space Upper Bound for Subset Sum}
Our second main result is~an~improved space upper bound
for \problem{Subset Sum}. We~achieve this by~improving and simplifying the algorithm by~Nederlof and W\k{e}grzycki.

\begin{theorem}
    \label{theorem:main_theorem}
    There exists a~Monte Carlo algorithm with constant 
	success probability that solves \problem{Subset Sum} for instances
    with $n$~integers of~absolute value at~most $2^{O(n)}$
    in~time $O^*(2^{0.5n})$ and space~$O^*(2^{0.246n})$.
\end{theorem}

The algorithm by~Nederlof and W\k{e}grzycki proceeds 
roughly as~follows.
Let $I \subseteq \mathbb Z$, $|I|=n$, be an~instance of~\problem{Subset Sum}. Assume that it~is 
a~yes-instance and that a~solution $S \subseteq I$ has size $n/2$. Assume further that there exists
a~subset  $M \subseteq I$ 
of~size $\Theta(n)$ such that 
$|S \cap M|=|M|/2$ and $M$~is a~perfect mixer: the weights 
of~all subsets of~$M$ are (pairwise) distinct. 
Partition $I \setminus M$ into two parts $L$~and~$R$ of~equal size, hence, $I=L \sqcup M \sqcup R$. We~will be~looking for 
two disjoint sets $S_1 \subseteq L \sqcup M$
and $S_2 \subseteq M \sqcup R$ such that $|S_1 \cap M|=|S_2 \cap M|=|M|/4$ and $S_1 \sqcup S_2$ is a~solution.
This problem can be~solved by~a~reduction to the Weighted Orthogonal Vectors (WOV) problem that can be~viewed as~a~hybrid of~the Orthogonal Vectors and the $2$-\SUM{} problems: given two families {$\mathcal L$ and $\mathcal R$} of~weighted sets {and a~target integer~$t$}, find two disjoint sets whose sum of~weights is~equal to~{$t$}. A~naive such reduction would proceed similarly to~the algorithm by~Horowitz and Sahni~\cite{DBLP:journals/jacm/HorowitzS74}: let 
\begin{align*}
\mathcal L&=\{(S_1 \cap M, w(S_1)) \colon S_1 \subseteq L \sqcup M, |S_1 \cap M|=|M/4|\},\\
\mathcal R&=\{(S_2 \cap M, w(S_2)) \colon S_2 \subseteq M \sqcup R, |S_2 \cap M|=|M/4|\}.
\end{align*} 
Then, it~remains to~solve WOV for {$\mathcal L$, $\mathcal R$ and $t$}.
The bad news is that $|L \sqcup M|>n/2$, hence enumerating all $S_1 \subseteq L \sqcup M$ leads to a~running time worse than $O^*(2^{n/2})$. The good news is that the solution~$S$
is~represented by~exponentially many pairs $(S_1, S_2)$: 
there are many ways to~distribute $S \cap M$ between $S_1$~and~$S_2$. This allows one to~use the representation technique: take a~random prime~$p$ of~the order $2^{|M|/2}$
and a~random remainder~$r \in \mathbb{Z}_p$; 
filter 
$\mathcal L$~and~$\mathcal R$ to~leave only subsets $S_1$~and~$S_2$ such that $w(S_1) \equiv_p r$ and $w(S_2) \equiv_p {t - r}$. This way, one reduces the space complexity since
$\mathcal L$~and~$\mathcal R$ become about~$p$ times smaller. 

In this work, we further introduce the following idea to reduce the space complexity of~the algorithm and to~simplify~it at~the same time.
We~use another prime~$q$ to~further filter~$\mathcal L$ and~$\mathcal R$: this allows~us to~reduce the original instance
to~a~larger number of~instances that in~turn have smaller size.
The idea is~best illustrated by the following toy example.
Consider an~instance~$I$ of~size~$n$ of \problem{Subset Sum}.
As~in~the algorithm by~Horowitz and Sahni~\cite{DBLP:journals/jacm/HorowitzS74}, partition~$I$ arbitrarily into two parts~$L$ and~$R$ of~size~$n/2$. Assume now that we~are given a~magic integer~$q=\Theta(2^{n/4})$ that hashes all subsets of~$L$~and~$R$ almost perfectly: for each $r \in \mathbb Z_q$,
\[|\{A \subseteq L \colon w(A) \equiv_q r\}|,\, |\{B \subseteq R \colon w(B) \equiv_q {t - r}\}| =O(2^{n/4}) \, .\]
This allows~us to~solve~$I$ in~time $O^*(2^{n/2})$ and space~$O^*(2^{n/4})$ without using a~reduction to~$4$-\SUM{} and priority queues machinery of~Schroeppel and Shamir's algorithm: for each~$r$, construct sequences of~all subsets of~$L$ and~$R$
having weights~$r$ and~{$t - r$} modulo~$q$, respectively, and solve
$2$-\SUM{} for the two resulting sequences of~length $O(2^{n/4})$.
Thus, instead of~reducing~$I$ to a~single instance of~$2$-\SUM{}
of~size $2^{n/2}$, we~reduce~it to~$2^{n/4}$ instances of~size~$2^{n/4}$.
We~show that even though in~general a~randomly chosen integer~$q$ does not give a~uniform distribution of~remainders,
one can still use this idea
to~reduce the space.

On~a~technical side, we~choose a~mixer~$M$
more carefully. This allows~us to~cover various corner cases, when the size of~the solution~$S$ is not $n/2$ and when 
the mixer~$M$~is far from being perfect, in~the same manner,
thus
simplifying the algorithm 
by~Nederlof and W\k{e}grzycki.

\section{General Setting}

\subsection{Modular Arithmetic}
We~write $a \equiv_m b$
to~denote that integers $a$~and~$b$ have the same remainder modulo~$m$.
By~$\mathbb{Z}_m$ we~denote the ring of~remainders modulo~$m$. We~make use of~the following well known fact.
\begin{theorem}[Chinese Remainder Theorem] \label{theorem:chinese_remainder_theorem}
    For any coprime positive integers $p$~and~$q$
    and any $a \in \mathbb Z_p$ and $b \in \mathbb Z_q$, there exists a~unique
    $x \in \mathbb Z_{pq}$ such that $x \equiv_p a$ and $x \equiv_q b$. 
    Furthermore,
    it~can be~computed (using the extended Euclidean algorithm) in~time $O(\log^2 p+\log^2 q)$.
\end{theorem}

\subsection{Prime Numbers}\label{section:primes}

Given an~integer~$t$, one can generate a~uniform random prime from $[t, 2t]$ in time $O(\log^{O(1)} t)$: to~do this, one selects a~random integer from $[t, 2t]$ and checks whether 
it~is prime~\cite{lp19}; the expected number of~probes is~$O(\log t)$ due to~the law of distribution of prime numbers.
We also use the following estimate: any positive integer~$k$ has at~most $\log_2 k$ distinct prime divisors; 
hence the probability that a~random prime from $[t, 2t]$ divides~$k$ is~at~most $O(\frac{\log k \log t}{t})$.

\subsection{Probability Amplification}\label{sec:amplification}
We~use the following standard success probability amplification trick:
if~a ``good'' event~$A$ occurs with probability at~least~$p$,
then for~$A$ to~occur at~least once with constant probability,
$1/p$ independent repetitions are sufficient:
\((1-p)^{1/p} \le (e^{-p})^{1/p}=e^{-1} \, .\)

\subsection{Growth Rate and Entropy}
Much like $O(\cdot)$ hides constant factors, $O^*(\cdot)$
and $\widetilde{O}(\cdot)$ hide factors that grow {polynomially in the input size} 
and polylogarithmically, respectively: for example, $n^{2}\cdot 1.5^n=O^*(1.5^n)$ and $\log^{3}n \cdot n^2=\widetilde{O}(n^2)$.
$\Omega^*(\cdot)$, $\Theta^*(\cdot)$, $\widetilde{\Omega}(\cdot)$, and $\widetilde{\Theta}(\cdot)$ are defined similarly.

We~use the following estimates for binomial coefficients: for any constant $0 \le \alpha \le 1$,
\begin{equation}\label{eq:binomialestimate}
\Omega(n^{-1/2})2^{h(\alpha)n} \le \binom{n}{\alpha n} \le 2^{h(\alpha)n} \text{ and hence } \binom{n}{\alpha n}=\Theta^*\left( 2^{h(\alpha)n}\right),
\end{equation}
where
\(h(x)=-x \log_2 x - (1-x)\log_2 (1-x)\)
is the binary entropy function. The function~$h$
is~concave: for any $0 \le \beta \le 1$ and any $0 \le x, y \le 1$,
\(\beta h(x) + (1-\beta) h(y) \le h(\beta x + (1-\beta) y) \, .\)

\subsection{Sets and Sums} 
For a~positive integer~$n$, $[n]=\{1, 2, \dotsc, n\}$.
For disjoint sets $A$~and~$B$, by $A \sqcup B$ we~denote their union.

Let $S \subseteq \mathbb Z$ be a~finite set of~integers. By~$w(S)$ we denote the weight of~$S$: $w(S)=\sum_{a \in S}a$. By $2^S$ we~denote the power set of~$S$, i.e., the set of~its subsets: 
$2^S=\{A \colon A \subseteq S\}$. By $\binom{S}{k}$
we~denote the set of all subsets of~$S$ of~size~$k$:
$\binom{S}{k}=\{A \subseteq S \colon |A|=k\}$.
For a~family of~sets $\mathcal F \subseteq 2^S$, 
by~$w(\mathcal F)$ we~denote the set of~their weights: 
$w(\mathcal F)=\{w(A) \colon A \in \mathcal F\}$. 

Given~$S$, one can compute $|w(2^S)|$
in~time and space~$O^*(2^{|S|})$ by~iterating over all of~its subsets. If $|w(2^S)|=2^{(1-\varepsilon)|S|}$, 
we~call~$S$ an~$\varepsilon$-mixer.

\begin{lemma}
\label{lemma:same_fractions}
    Let $S \subseteq A$ and $|A|=n$.
    For $m=O(1)$ uniformly randomly chosen disjoint non-empty sets $A_1, \dotsc, A_m \subseteq A$, 
	\[\Pr\left[\frac{|A_i \cap S|}{|A_i|} \approx \frac{|S|}{|A|} \text{, for all $i \in [m]$}\right] = \Omega^*(1) \; ,\]
 where $\frac{|A_i \cap S|}{|A_i|} \approx \frac{|S|}{|A|}$ means that $\left||A_i \cap S| - |A_i| \frac{|S|}{|A|}\right| = O(1)$.
\end{lemma}

 \begin{proof}\label{proof:same_fractions}

 	We will show that

         \[\Pr\left[|A_i \cap S| = \left\lfloor |A_i| \frac{|S|}{|A|}\right\rfloor \, \forall i \in [m - 1] \land |A_m \cap S| = \left \lfloor |A_m| \frac {|S| - \sum_{i < m} \lfloor A_i \frac {|S|} {|A|}\rfloor} {|A| - \sum_{i < m} |A_i|} \right\rfloor \right] = \Omega^*(1) \; .\] 

     After that it is easy to see that $\left||A_i \cap S| - |A_i| \frac{|S|}{|A|} \right| = O(1), \; \forall i \in [m]$.

     Hence, we can think that all numbers $|A_i| \frac{|S|}{|A|}$ are integers.
   
     Since $m=O(1)$, it~suffices to~estimate the probability for~$m=1$.
 	{Let $B \subseteq A$~be a~random set and 
 	\begin{equation*}
 		\gamma = \frac {\left\lfloor |B| \cdot \frac{|S|}{|A|} \right\rfloor} {|B|}.
 	\end{equation*}
 	}
     Then
     \begin{align*}
 		\Pr\left[|B \cap S|=\left\lfloor|B| \frac {|S|} {|A|}\right\rfloor\right] &= \frac{\binom{|S|}{\gamma |B|} \binom{|A \setminus S|}{(1 - \gamma) |B|}}{\binom{|A|}{|B|}}\tag{by~\eqref{eq:binomialestimate}}\\
         &= \Omega^*\left(2^{\gamma |A| h(|B|/|A|) + (1 - \gamma)|A| h(|B|/|A|) - |A| h(|B|/|A|)}\right)
         =\Omega^*(1).
     \end{align*}
 \end{proof}

\subsection{Weighted Orthogonal Vectors}

\begin{restatable}{definition}{WOV}[Weighted Orthogonal Vectors, WOV] \label{definition:wov} 
    Given families of $N$ weighted sets $\mathcal A, \mathcal B \subseteq 2^{[d]} \times \mathbb N$, and a target integer~$t$, 
    the problem is to find $(A, w_A) \in \mathcal A$ and $(B, w_B) \in \mathcal B$ such that $A \cap B = \emptyset$ and $w_A + w_B = t$.
\end{restatable}

\begin{restatable}{lemma}{lemmawov}
\label{lemma:WOV}
    For any $\sigma \in [0, \frac 1 2]$, there is a Monte-Carlo algorithm that, given $\mathcal A \subseteq \binom{[d]}{\sigma d} \times \mathbb N$, $\mathcal B \subseteq \binom{[d]}{(1/2-\sigma)d} {\times \mathbb N}$  and a~target integer~$t$, 
    solves $\operatorname{WOV}(\mathcal A, \mathcal B)$
    in time
$
    \widetilde O\left((|\mathcal A| + |\mathcal B|) 2^{d(1-h(1/4))}\right)
$
and space $\widetilde O\left(|\mathcal A| + |\mathcal B| + 2^d\right)$.
\end{restatable}

To~prove the lemma, we utilize the algorithm from~\cite{DBLP:conf/stoc/NederlofW21}, but with a~different parameter range.
See proofs in section \ref{sec:wov_lemmas}.

\section{New Arthur--Merlin Algorithm}

\foursumtheorem*
\begin{proof}
	Given arrays $A_1, A_2, A_3, A_4$ each consisting~of $n$~$4\log n$-bit integers,
	our goal is~to~find four indices $i_1, i_2, i_3, i_4 \in [n]$ such that
	\(A_1[i_1]+A_2[i_2]+A_3[i_3]+A_4[i_4]=0 \, .\)
	Without loss of~generality, assume that $A_i$'s are sorted 
	and do~not contain duplicates: for every $i \in [4]$,
	\(A_i[1] < A_i[2] < \dotsb < A_i[n] \, .\)
	Assume that $i_1^*, i_2^*, i_3^*, i_4^*$
	is a~solution and let $s=A_1[i_1^*]+A_2[i_2^*]$.
	
	Let $\mathcal P$~be a~randomized algorithm
	that, given an~integer~$n$, returns a~uniform random prime from $[2..n]$. Since one can check whether a~given integer~$k$ is~prime or~not in~time~$O(\log^{O(1)}k)$, the expected running time of~$\mathcal P$ is $O(\log^{O(1)}n)$: one picks a~random $k \in [2..n]$ and checks whether 
	it~is prime; if it~is not, one repeats; the expected
	number of~probes is~$O(\log n)$.
	
	We~are ready to~describe the protocol.	
	\begin{quote}
	\begin{description}
		\item[Stage~0: Shared randomness.]
 		Using the shared randomness, Arthur and Merlin generate a~random prime $p=\mathcal{P}(n)$.
		\item[Stage~1: Preparing a~proof.]
		Merlin sends an~integer $0 \le r < p$ to~Arthur.
		\item[Stage~2: Verifying the proof.] 
		Arthur takes all elements of~$A_1$ and~$A_2$ modulo~$p$
		and sorts both arrays. Using the two pointers method (or~binary search), 
		he~enumerates all pairs $(i_1,i_2) \in [n]^2$ such that
		$(A_1[i_1]+A_2[i_2]) \equiv r \bmod p$. For each such pair $(i_1, i_2)$,
		he~adds $A_1[i_1]+A_2[i_2]$ to a~new array $A_{12}$.
		If~the size of~$A_{12}$ becomes larger than $O(n\log^3 n)$, Arthur rejects immediately.
		Then, he~does the same for $A_{3}$~and~$A_4$, but for the remainder $-r$ (rather than~$r$)
		and creates an~array~$A_{34}$.
		Finally, he~solves $2$-\SUM{} for $A_{12}$ and $A_{34}$ (the two arrays may have different lengths, but this can be~fixed easily by~padding one of them with dummy elements).
	\end{description}
	\end{quote}
	
	Now, we~analyze the protocol. It~is~clear that the proof size~is at~most~$\log_2 n$ and that the running time of~the verification stage is~$\widetilde{O}(n)$. Also, if~there 
	is~no~solution for the original instance, Arthur rejects with probability~$1$, as~he only accepts if~the solution is~found. 
	Below, we~show that if~Merlin sends $r=s \bmod p$ (recall that $s=A_1[i_1^*]+A_2[i_2^*]$),
	Arthur accepts with good enough probability.
	
	Consider the set
	\(S=\{(i_1, i_2) \in [n]^2 \colon A_1[i_1]+A_2[i_2] \equiv s \bmod p\}\)
	(hence, $|A_{12}| =|S|$). It~contains at~most~$n$ true positives, that is,
	pairs $(i_1,i_2)$ such that $A_1[i_1]+A_2[i_2]=s$: as~the arrays do~not
	contain duplicates, for every $i_1$ there is~at~most one matching~$i_2$.
	All the other pairs $(i_1,i_2)$ in~$S$ are false positives:
	\(A_1[i_1]+A_2[i_2]-s \neq 0\),  but \(A_1[i_1]+A_2[i_2]-s \equiv 0 \bmod p\).
	Recall that $|A_1[i_1]+A_2[i_2]-s|=O(n^4)$, hence the number of~prime divisors
	of~$A_1[i_1]+A_2[i_2]-s$ is~$O(\log n)$, Section \ref{section:primes}. The probability that a~random prime $2 \le p \le n$ 
	divides $(A_1[i_1]+A_2[i_2]-s)$ is $O(\frac{\log^{2}n}{n})$.
	Thus, the expected number of~false positives is~at~most
	\(n^2 \cdot O\left(\frac{\log^2n}{n}\right) = O(n\log^2 n)\).
	By~Markov's inequality, the probability that the number of~false positives
	is~larger than $O(n\log^{3}n)$ (and hence that $|S| = O(n\log^3 n)$) is at~most $1/\log n$.

	Thus, with probability at~least $1-2/\log n$, both
	$A_{12}$ and $A_{34}$ have length at~most $O(n\log^3n)$
	and by~solving $2$-\SUM{} for them, Arthur finds a~solution for the original yes-instance.
\end{proof}

\cktsattheorem*
\begin{proof}
	Assume that, for some $\varepsilon>0$, there exists
	an~algorithm~$\mathcal A$ that checks whether a~given Boolean circuit
	with $g$~binary gates and $n$~inputs is~satisfiable
	in~time $O(g2^{(1-\varepsilon)n})$. Using~$\mathcal A$,
	we~will solve $4$-\SUM{} in time $O(n^{2-\varepsilon})$.
	
	Consider the verification algorithm from \Cref{thm:foursum} as a~Turing machine~$M$:
	$M$~has three read-only input tapes $(I_R, I_P, I)$: the tape $I_R$ contains random bits, the tape $I_P$ contains a~proof, and the tape $I$ contains four input arrays.
	For any four arrays $(A_1,A_2,A_3,A_4)$ of size~$n$,
	if it~is a~yes-instance, there exists a~proof
	of~size at~most $\log_2n$ such that $M$~accepts
	with probability at~least $2/3$, whereas for a~no-instance, $M$~rejects every proof with probability~$1$. 
	The machine~$M$ compares $2\log n$-bit integers and moves two pointers 
	through two arrays.
	Hence, the running time of~$M$ is~$\widetilde{O}(n)$.
	
	Given~$(A_1,A_2,A_3,A_4)$, fix the contents of the tape~$I$ and fix the random bits of the tape $I_R$.
	This turns~$M$ into a~machine~$M'$ with $\log_2 n$ input 
	bits and running time~$\widetilde{O}(n)$.
	
	As~proved by~\cite{DBLP:journals/jacm/PippengerF79}, 
	a~Turing machine recognizing a~language~$L \subseteq \{0,1\}^*$ with running time~$t(n)$
	can be~converted to an~infinite series $\{C_n\}_{n=1}^{\infty}$ of~circuits of~size $O(t(n)\log t(n))$:
	for every~$n$, $C_n$ has $n$~inputs and $O(t(n)\log t(n))$ gates
	and computes a~function $f \colon \{0,1\}^n \to \{0,1\}$ such that
	$f^{-1}(1)=L \cap \{0,1\}^n$. The circuit~$C_n$
	can be~produced in~time $O(t(n)\log t(n))$.
	
	Applying this to~the machine~$M'$,
	we~get a~circuit~$C$ with $\widetilde{O}(n)$
	gates and $\log_2 n$ input bits.
	If $A_1, A_2, A_3, A_4$ is a~yes-instance of~$4$-\SUM{}, the circuit~$C$ is satisfiable with 
	probability at~least~$2/3$; otherwise it~is unsatisfiable. Using the algorithm~$\mathcal A$,
	we~can check the satisfiability of~$C$ in~time $\widetilde{O}(n2^{(1-\varepsilon)\log_2n})=\widetilde{O}(n^{2-\varepsilon})$.
\end{proof}

\subsetsumtheorem*
\begin{proof}
	We~use the standard reduction from \problem{Subset SUM}
	to~$4$-\SUM{}.
	Given a~sequence~$A$ of~$n$~integers,
	partition them into four parts of~size $n/4$
	and create sequences $A_1,A_2,A_3,A_4$
	of~size $2^{n/4}$ containing sums of~all subsets
	of~the corresponding parts. Use the protocol from 
	\Cref{thm:foursum} for the resulting instance of~$4$-\SUM{}.
\end{proof}

\twoksumcorollary*

\begin{proof}
	Reduce $2k$-\SUM{} to $4$-\SUM{}: given sequences 
 	$A_1, \dotsc, A_{2k}$ of~length~$n$, 
 	partition them into four parts
 	each containing $k/2$ sequences of~length~$n$.
 	Let $B_1, B_2, B_3, B_4$ be~sequences of~length $n^{k/2}$ containing
 	all sums of~$k/2$ elements of~the corresponding part
 	(where we~take exactly one element from each sequence).
 	Use the protocol from 
 	\Cref{thm:foursum} for the resulting instance of~$4$-\SUM{}.
\end{proof}

\cktsattwokcorollary*
 \begin{proof}

 	Assume that \problem{Circuit SAT} can be~solved in~time 
 	$O(g^c2^{(1-\varepsilon)n})$ for constants~$c, \varepsilon>0$.
	
 	Similarly to~the proof of~\Cref{thm:cktsat}, given an~instance of~$2k$-\SUM{}, create a~circuit~$C$
 	with $k\log_2 n$ inputs and of~size $n^{o(k)}$
 	that solves this instance with one-sided error probability at~most~$1/3$.
	
 	Using the above algorithm, one can check the satisfiability of~$C$ in~time
 	\[O(n^{c\cdot o(k)}\cdot 2^{(1-\varepsilon)k\log_2 n})=O(n^{o(k)}n^{(1-\varepsilon)k}) \, .\]
 	For large enough~$k$, this allows to~solve $2k$-\SUM{}
 	in~time $O(n^{(1-\varepsilon/2)k})$.
 \end{proof}

\section{Improved Space Upper Bound for Subset Sum}
\subsection{Algorithm}

Let $(I, t)$~be an~instance of~\problem{Subset Sum} (as~usual, $n=|I|$)
and assume that it is a~yes-instance: we~will present
a~randomized algorithm that is~correct on~no-instances
with probability~1, so,~as~usual, the main challenge is~to
obtain $\Omega^*(1)$ success probability for yes-instances.
Let $S \subseteq I$ be a~solution: $w(S)=t$. 
It~is unknown 
to~us, but we~may assume that we~know its size: there are 
just~$n$ possibilities for~$|S|$, so~with a~polynomial overhead 
we~can enumerate all of~them. Thus, assume that $|S|=\alpha n$ 
where $0 < \alpha \le 1$ is known to~us.

For a~parameter $\beta=\beta(\alpha)<0.15$ to be~specified later, 
select randomly pairwise disjoint sets $M_L,M,M_R \in \binom{I}{\beta n}$. In~time and space $O^*(2^{0.15n})$, find $\varepsilon_L, \varepsilon, \varepsilon_R$ such that 
$M_L$~is an $\varepsilon_L$-mixer,
$M$~is an $\varepsilon$-mixer,
and
$M_R$~is $\varepsilon_R$-mixer.
Without loss of~generality, we~assume that 
$\varepsilon \le \varepsilon_L, \varepsilon_R$.
Consider the probability that $S$
touches exactly half of~$M_L \sqcup M \sqcup M_R$: 

\begin{align*}
\Pr\left[|(M_L \cup M \cup M_R) \cap S| = 3\beta n/2\right] =&  
    \frac{\binom{\alpha n}{3\beta n/2} \binom{(1 - \alpha) n}{3 \beta n/2}}{\binom{n}{3 \beta n}} = \Theta^*\left(2^{-\lambda n}\right) \\
\frac{\binom{\alpha n}{3\beta n/2} \binom{(1 - \alpha) n}{3 \beta n/2}}{\binom{n}{3 \beta n}} = 2^{-\lambda n} \iff&  \frac 1 n \log \frac{\binom{n}{3 \beta n}}{\binom{\alpha n}{3\beta n/2} \binom{(1 - \alpha) n}{3 \beta n/2}} = \lambda \iff \notag \\
\iff& \lambda = h(3 \beta) - \alpha h\left(\frac {3 \beta} {2 \alpha}\right) - (1 - \alpha) h\left(\frac {3 \beta} {2(1 - \alpha)}\right), \tag{by~\eqref{eq:binomialestimate}}
\end{align*}

hence,
\begin{equation}
    \lambda = h(3\beta) - \alpha h\left(\frac{3\beta}{2 \alpha}\right) - (1 - \alpha)h\left(\frac{3\beta}{2(1 - \alpha)}\right) . \label{eq:lambda}
\end{equation}

By~repeating this process $2^{\lambda n}$ times, we~ensure that the event
$|(M_L \cup M \cup M_R) \cap S| = 3\beta n/2$ 
happens with probability $\Omega^*(1)$ (recall 
Section~\ref{sec:amplification}).
Further, we~may assume that 
\[|M_L \cap S| = |M \cap S| = |M_R \cap S| = \beta n/2, \] as by 
Lemma~\ref{lemma:same_fractions}, $|M_L \cap S| = |M \cap S| = |M_R \cap S| = \beta n / 2 + O(1)$ with probability $\Omega^*(1)$ conditioned on~the event that $S$~touches {$3\beta n / 2$ elements of} $M_L \sqcup M \sqcup M_R$.
An~exact equality can be~assumed since we~can search for an~exact value in the neighbourhood of~$\beta n / 2$ and the constant difference will not affect the memory and time complexity.
In~the following applications of 
Lemma~\ref{lemma:same_fractions}, we~omit $O(1)$ factors for similar considerations.

We~argue that either~$S \cap M$ or~$M \setminus S$ is an~$(\le \varepsilon)$-mixer. Indeed, $M$~is an~$\varepsilon$-mixer and hence $|w(2^M)|=2^{(1-\varepsilon)|M|}$. Now, $|S \cap M|=|M \setminus S|=|M|/2$ and
\(|w(2^{S \cap M})| \cdot |w(2^{M \setminus S})| \ge | w(2^M)| \, .\)
Hence,
\[\max\left\{ |w(2^{S \cap M})| , |w(2^{M \setminus S})|\right\} \ge 2^{(1-\varepsilon) |M|/2},\]
implying that either~$S \cap M$ or~$M \setminus S$ is indeed a~$(\le \varepsilon)$-mixer. 
In~the following,
we~assume that it~is $S \cap M$ that is~$(\le \varepsilon)$-mixer. To~assume this without loss of~generality, 
we~run the final algorithm twice~--- for $(I,t)$
and $(I,w(I)-t)$.

For any (finite and non-empty) set~$A$, there exists $1 \le k \le |A|$, such that $|w(\binom{A}{k})| \ge |w(2^A)|/|A|$
(by~the pigeonhole principle). {Since $|w(\binom{A}{k})| = |w(\binom{A}{|A| - k})|$, we can even assume that $1 \le k \le |A|/2$.} Consider the 
corresponding~$k$ for the set~$A=S \cap M$ and let $\mu=k/(2|S \cap M|)$ 
(hence, $0 \le \mu \le { 0.25})$.
Then,
\begin{equation}
    \label{eq:mu}
    \left|w\left(\binom{S \cap M}{2 \mu |S \cap M|}\right)\right| \ge \frac{\left|w\left(2^{S \cap M}\right)\right|}{|S \cap M|} \, .
\end{equation}

Partition $I \setminus (M_L \sqcup M \sqcup M_R)$ randomly into 
$L_1 \sqcup L_2 \sqcup L_3 \sqcup L_4 \sqcup R_1 \sqcup R_2 \sqcup R_3 \sqcup R_4$ where the exact size of~all eight parts
will be specified later. Let $L = L_1 \sqcup L_2 \sqcup L_3 \sqcup L_4$ and $R = R_1 \sqcup R_2 \sqcup R_3 \sqcup R_4$.
By~Lemma~\ref{lemma:same_fractions}, with probability $\Omega^*(1)$ the set~$S$ covers the same fraction of~each of $L_i$'s and $R_i$'s. We~denote this fraction by~$\gamma$:
\begin{equation}
    \gamma=\frac{|(L \sqcup R) \cap S|}{|L \sqcup R|}=\frac{\alpha-3\beta/2}{1-3\beta} \, . \label{eq:gamma_definition}
\end{equation}

Now, we~apply the representation technique. 
Note that there exist many ways to~partition~$S$
into $S_L \sqcup S_R$ such that
\begin{equation}\label{eq:partition}
    S_L \subseteq M_L \sqcup L \sqcup M 
    \text{ and }
    S_R \subseteq M \sqcup R \sqcup M_R.
\end{equation} 
Indeed, the elements from $S \cap M$ can be~distributed arbitrarily between $S_L$~and~$S_R$.
Below, we~show that for a~large random prime number~$p$
and a~random remainder~$r$ modulo~$p$, 
the probability that there exists
a~partition $S=S_L \sqcup S_R$ such that $w(S_L) \equiv_p r$
is~$\Omega^*(1)$ (informally, at~least one representative $S_L \sqcup S_R$ of~the original solution~$S$ survives, even if~we ``hit'' all such partitions with a~large prime).

\begin{lemma}
    \label{lemma:p1}
    With constant probability over choices 
    of a prime number~$p$ such that $p \in [X/2, X]$, where $X \leq |w(2^{S \cap M})|$, the number of~remainders~$a \in \mathbb{Z}_p$ such that there exists a~subset $M' \subseteq S \cap M$ with $|M'|=2\mu |S \cap M|$ and $w(M') \equiv_p a$ is~close to~$p$:
    \[
    \Pr\left[ \left| \left\{ a \in \mathbb{Z}_p \colon \text{there exists } M' \subseteq S \cap M, |M'| = 2\mu |S \cap M|, w(M') \equiv_p a \right\} \right| \ge \Omega\left(p/n^2 \right) \right] \ge 9/10 .
    \]
\end{lemma} 
 \begin{proof} \label{proof:p1}
     Let $Q = M \cap S$. Then $|Q| = \frac \beta 2 n$.
     Recall that $|w(\binom{Q}{2 \mu |Q|})| \ge |w(2^{Q})| / |Q|$ and 
     $p$ is a random prime in $[X/2, X]$.

     Let $\mathcal M = \binom{Q}{2 \mu |Q|}$, $c_i = | \{ M \colon M \in \mathcal M, w(M) \equiv_p i \} |$, and $d_i = [c_i > 0]$.
     We need to show that
 \(
     \Pr[ | \{ i \colon d_i = 1 \} | \ge \Omega(p/n^2) ] \ge 9/10.
 \)
     The expected number of collisions is
 \[
     \mathbb E\left [ \sum_i c_i^2 \right] = \sum_{M_1, M_2 \in \mathcal M} \Pr\left[ p \text{ divides } w(M_2) - w(M_1) \right] \le |\mathcal M| + O(n^2 |\mathcal M|^2 / X) .
 \]
     The last inequality follows from Section~\ref{section:primes}, all numbers are not greater than $2^{O(n)}$.
     Hence, 
     \(\sum_i c_i^2 = O(|\mathcal M| + n^2 |\mathcal M|^2 / X) = O( n^2|\mathcal M|^2 / p) 
     \) with constant probability. Then, the Cauchy--Schwartz inequality implies that
 \[
     |\mathcal M| = \sum_i c_i = \sum_i c_i d_i \le \sqrt{\sum_i d_i^2} \sqrt{\sum_i c_i^2} = \sqrt{\sum_i d_i} \sqrt{\sum_i c_i^2} \, .
 \]
     Hence,
 \[
     \sum_i d_i \ge \frac{(\sum_i c_i)^2}{\sum_i c_i^2} \ge \frac{|\mathcal M|^2}{O(n^2 |\mathcal M|^2 / p)} \ge \Omega(p/n^2) \, .
 \]
 \end{proof}

Take a~random prime $2^{\frac {1 - \varepsilon} 2 \beta n - 1} \le p \le 2^{\frac {1 - \varepsilon} 2 \beta n}$ and a~random remainder $r \in \mathbb Z_p$. 
Lemma~\ref{lemma:p1} ensures that with probability $\Omega^*(1)$, there exists $M' \subseteq S \cap M$ such that 
$|M'| = 2\mu|S \cap M|=\mu \beta n$ and $w(M') \equiv_p r$.

To~find the partition $S_L \sqcup S_R$, we~are going to~enumerate all sets $A \subseteq M_L \sqcup L \sqcup M$ with $w(A) \equiv_p r$ and $|A \cap (M \cap S)|=\mu\beta n$ 
as~well as~all sets $B \subseteq M \sqcup R \sqcup M_R$
with $w(B) \equiv_p t-r$ and 
$|B \cap (M \cap S)|=(1/2-\mu)\beta n$. Then, it~suffices to~solve WOV for the sets of~$A$'s and $B$'s.

To~actually reduce to~WOV, we~first construct set families $\mathcal M_L \subseteq 2^{M_L}$ and $\mathcal M_R \subseteq 2^{M_R}$ such that $|\mathcal M_L|=|w(2^{M_L})|$ and 
$|\mathcal M_R|=|w(2^{M_R})|$. That is, $\mathcal M_L$ contains
a~single subset of~$M_L$ for each possible weight. One can construct $\mathcal M_L$ (as~well as~$\mathcal M_R$) simply by~going through all subsets of~$M_L$ and checking, for each subset, whether its sum is~already in~$\mathcal M_L$.

Then, construct the following families (see Figure~\ref{figure:partition}).
    \begin{align}
	\mathcal Q_1 &= \{M' \cup L' \colon M' \in \mathcal M_L, L' \subseteq L_1, |L'| = \gamma |L_1| \},\label{eq:q_1_definition}\\
	\mathcal Q_2 &= \{L' \colon L' \subseteq L_2, |L'| = \gamma |L_2| \}, \label{eq:q_2_definition}\\
	\mathcal Q_3 &= \{L' \colon L' \subseteq L_3, |L'| = \gamma |L_3| \}, \label{eq:q_3_definition}\\
	\mathcal Q_4 &= \{L' \cup M' \colon L' \subseteq L_4, M' \subseteq M, |L'| = \gamma |L_4|, |M'| = \mu \beta n \}, \label{eq:q_4_definition}\\
        \nonumber\\
	\mathcal Q_1' &= \{M' \cup R' \colon M' \in \mathcal M_R, R' \subseteq R_1, |R'| = \gamma |R_1| \}, \label{eq:q_1'_definition}\\
	\mathcal Q_2' &= \{R' \colon R' \subseteq R_2, |R'| = \gamma |R_2| \}, \label{eq:q_2'_definition}\\
	\mathcal Q_3' &= \{R' \colon R' \subseteq R_3, |R'| = \gamma |R_3| \}, \label{eq:q_3'_definition}\\
	\mathcal Q_4' &= \{R' \cup M' \colon R' \subseteq R_4, M' \subseteq M, |R'| = \gamma |R_4|, |M'| = (1/2 - \mu) \beta n \}. \label{eq:q_4'_definition}
    \end{align}

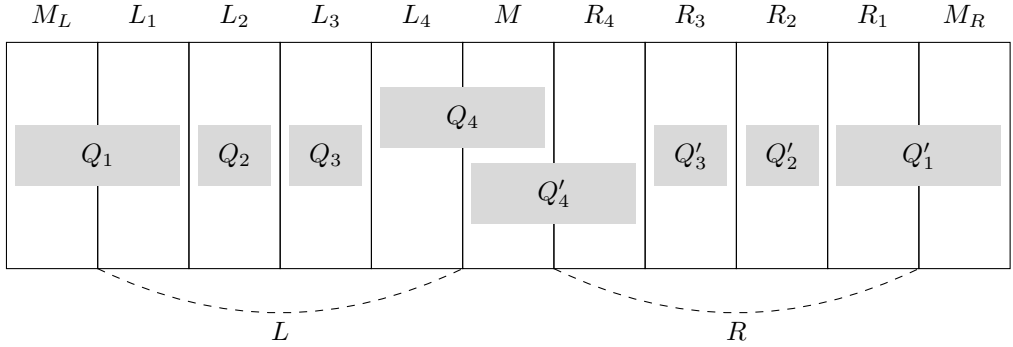
\begin{figure}[!ht]
\begin{center}
\begin{tikzpicture}[xscale=1.2]
	\foreach \i in {0,...,10}
		\draw (\i, 0) rectangle (\i + 1, 3);
	\foreach \l [count=\i] in {M_L, L_1, L_2, L_3, L_4, M, R_4, R_3, R_2, R_1, M_R}
		\node[above] at (\i - .5, 3) {\strut $\l$};
	\path (1, 0) edge[bend right=30, dashed] node[below] {$L$} (5, 0);
	\path (6, 0) edge[bend right=30, dashed] node[below] {$R$} (10, 0);
	\foreach \x/\y/\w/\l in {1/1.5/2/Q_1, 2.5/1.5/1/Q_2, 3.5/1.5/1/Q_3, 5/2/2/Q_4, 6/1/2/Q_4',
	7.5/1.5/1/Q_3', 8.5/1.5/1/Q_2', 10/1.5/2/Q_1'} {
		\filldraw[draw=none, fill=gray!30] (\x - \w/2 + .1, \y-.4) rectangle (\x + \w/2 - .1, \y + .4);
		\node at (\x, \y) {$\l$};
	}
\end{tikzpicture}
\end{center}
\caption{Partition of~the instance~$I$ (white) and its~solution~$S$ (gray)
into parts.}
\label{figure:partition}
\end{figure}

    Let $\mathcal Q_{max}$ stand for $\mathcal{Q}_i$ or $\mathcal{Q}_i'$ with the maximum size.   
    For $a, p \in \mathbb{Z}_{>0}$, let
    \[\mathcal{X}_{a} = \left\{ (Q_1, Q_2, Q_3, Q_4) \colon Q_i \in \mathcal{Q}_i, \sum w(Q_i) = a\right\} \text{ and } \mathcal{X}_{a, p} = \bigcup\limits_{a' \equiv_{p} a} \mathcal{X}_{a'} .\]
    Then, it~remains to~solve WOV for $\mathcal{X}_{r, p}$ and $\mathcal{X'}_{(t - r), p}$.
    To~do this, we~use the idea used in the Arthur--Merlin protocol: we reduce to 4-\SUM{} and use a~cutoff for the number of found candidates.
    The following lemmas estimate the size of~these set families.

\begin{lemma}
\label{lemma:divide_by_p1}
        With probability $\Omega^*(1)$ over the choice of~$r$, the size of $\mathcal X_{r, p}$ is at most $O^*(|\mathcal{Q}_{max}|^4 / p)$ and there exists $M_1 \subseteq M \cap S$ such that $w((S \cap (L \cup M_L)) \cup M_1) \equiv_p r$. 
\end{lemma}
 \begin{proof}
 		Let $A$ be an event of existing $M_1 \subseteq M \cap S$ such that $w((S \cap (L \cup M_L)) \cup M_1) \equiv_p r$. 
         Note that 
         \[ \mathbb{E}[\mathcal X_{r, p}] \leq |\mathcal{Q}_{max}|^4 / p \text{ and } \Pr[\mathcal X_{r, p} > x] \leq \frac{|\mathcal{Q}_{max}|^4}{p x} \, .\]
         As~$\Pr[A] = \Omega^*(1)$, we can take 
         $x = O^*(\frac{|\mathcal{Q}_{max}|^4}{p})$ in a~way that will make
         $\Pr[\mathcal X_{r, p} > x] \leq \frac{\Pr[A]}{2}$. 
         Then \[\Pr[\mathcal X_{r, p} \leq x \wedge A] \geq \Pr[A] - \Pr[\mathcal X_{r, p} > x] \geq \frac{\Pr[A]}{2} = \Omega^*(1) \,.\]
 \end{proof}

    To~further reduce the size of~these families, 
    we choose another random prime number 
        \[q = \Theta^*\left(\frac{2^{n/2}}{2^{\lambda n}|\mathcal{Q}_{max}|^2}\right),\]
    and iterate over all possible remainders $s \in \mathbb Z_{q}$.
    We aim to find two families of sets as the above ones, but with weights congruent to $w(S_L)$ and $(t - w(S_L))$, respectively, modulo~$p$ and~$q$.
    While iterating on~$s$, at some point we~try $s \equiv_{q} w(S_L)$, so we can assume that we guessed~$s$ correctly and construct the families for the fixed~$r$ and~$s$. 
    We now focus on the first family only, since we can then find the second one independently using the same algorithm. 

    For $a \in \mathbb{Z}_{>0}$, let \(\mathcal{L}_a = \{(a, Q_4 \cap M) \colon (Q_1, Q_2, Q_3, Q_4) \in \mathcal{X}_a\}\). For $a_1, a_2, p, q \in \mathbb{Z}_{>0}$, let \[\mathcal{X}_{a_1, p, a_2, q} = \bigcup\limits_{\substack{a \equiv_{p} a_1 \\ a \equiv_{q} a_2}} \mathcal{X}_a \text{ and }\mathcal{L}_{a_1, p, a_2, q} = \bigcup\limits_{\substack{a \equiv_{p} a_1 \\ a \equiv_{q} a_2}} \mathcal{L}_a .\] 

    Our goal is to find a family $\mathcal{L}$ containing $\mathcal{L}_{w(S_L)}$ in order to guarantee that we later find the solution using the WOV algorithm. But we need $\mathcal{L}$ to be small enough. Let $\ell = \frac{|\mathcal{Q}_{max}|^4}{p q} + \binom{|M|}{\max\{\mu, 1/2-\mu\} |M|}$ represent {an upper bound on the sum of the} average size of {$ \mathcal X_{a_1, p, a_2, q}$ and the size of $\mathcal L_{w(S_L)}$}. We want $\mathcal{L}$ to be of size $\widetilde{O}(\ell)$. 
    The following lemma shows that 
    $\mathcal{X}_{w(S_L), p, w(S_L), q}$
    contains not too many redundant sets.
    
    \begin{lemma}
        \label{lemma:q}
        With probability $\Omega^*(1)$ over the choice of~$q$, 
        \(|\mathcal{X}_{r,p, w(S_L),q} \setminus \mathcal{X}_{w(S_L)}| \le \ell .\)     
    \end{lemma}
     \begin{proof}
         We estimate the expected number of elements in $\mathcal \mathcal{X}_{r,p, w(S_L) ,q} \setminus \mathcal{X}_{w(S_L)}$. 
         To do that, we calculate the sum of $\Pr_{q}[\sum w(Q_i) \equiv_{q} w(S_L)]$ over all tuples $(Q_1, Q_2, Q_3, Q_4) \in \mathcal{X}_{r,p} \setminus \mathcal{X}_{w(S_L)}$ such that $\sum w(Q_i) \neq w(S_L)$. 

         Since
         $\Pr_{q}[\sum w(Q_i) 
         \equiv_{q} w(S_L)] = O^*(\frac{1}{q})$ (recall Section~\ref{section:primes}), \[\mathbb E |\mathcal{X}_{r, p, w(S_L), q} \setminus \mathcal{X}_{w(S_L)}| = |\mathcal{X}_{r,p} \setminus \mathcal{X}_{w(S_L)}| \cdot O^*(1).\] By Lemma~\ref{lemma:divide_by_p1}, it is equal to $\widetilde{O}({|\mathcal{Q}_{max}|^4}/ (pq)) = \widetilde{O}(\ell)$.         
     \end{proof}

    Let us find the first $\le 2\ell + 1$ elements of $\mathcal{X}_{r, s}$ using the algorithm from Lemma \ref{lemma:shamir_for_rings} which runs in time $\widetilde{O}(|\mathcal{Q}_{max}|^2 + \ell)$ and uses $\widetilde{O}(|\mathcal{Q}_{max}| + \ell)$ space.

    \begin{lemma}
    \label{lemma:shamir_for_rings}
        There is an algorithm that given $\mathcal{Q}_1, \mathcal{Q}_2, \mathcal{Q}_3, \mathcal{Q}_4$, $(r, p, s, q)$ and integer $m$, finds $\min(|\mathcal{X}_{r, p, s, q}|, m)$ elements of $\mathcal{X}_{r, p, s, q}$ using $\widetilde{O}(|\mathcal{Q}_{max}|^2 + m)$ time and $\widetilde{O}(|\mathcal{Q}_{max}| + m)$ space.
    \end{lemma}
     \begin{proof}
         Let $v \in [0, pq - 1]$ be an integer such that $v \equiv_{p} r$ and $v \equiv_{q} s$. We obtain such $v$ in $\Omega^*(1)$ time by Theorem \ref{theorem:chinese_remainder_theorem}. 
         Our goal is to output every tuple $(Q_1, Q_2, Q_3, Q_4)$ such that $Q_i \in \mathcal{Q}_i$ and $\sum w(Q_i) \equiv_{pq} v$, one by one and halt the algorithm when the size of the output becomes $m$. 

         We can assume that all weights of sets from $\mathcal{Q}_i$ are in $[0, pq - 1]$.
         Then, we want to output the tuples with total weights equal to $v$, $v + pq$, $v + 2pq$ or $v + 3pq$. Now we run Schroeppel and Shamir's~\cite{DBLP:journals/siamcomp/SchroeppelS81} algorithm four times keeping the total size of the output at most $m$.
     \end{proof}

    Assume that $s \equiv_{q} w(S_L)$.
    If the algorithm yields at most $2\ell$ elements then we know the whole $\mathcal{X}_{r, p, w(S_L), q}$, so we can now just construct {$\mathcal{L}$ from $\mathcal{X}_{r, p, s, q}$ in $\widetilde{O}(\ell)$ time and space.}
    Otherwise, by Lemma \ref{lemma:q}, we have that $|\mathcal{X}_{w(S_L)}| > \ell$ with probability $\Omega^*(1)$. 
    Since strictly more than half of the outputted tuples have weight exactly $w(S_L)$, we can determine $w(S_L)$ in $\widetilde{O}(\ell)$ time. When we know $w(S_L)$, we use Lemma \ref{lemma:faster_shamir} to obtain $\mathcal{L}_{w(S_L)}$ using $\widetilde{O}(|\mathcal{Q}_{max}|^2)$ time and $\widetilde{O}(|\mathcal{Q}_{max}|)$ space.

    \begin{lemma}
    \label{lemma:faster_shamir}
        There is an algorithm that given $\mathcal{Q}_1, \mathcal{Q}_2, \mathcal{Q}_3, \mathcal{Q}_4$ and $a$, finds $\mathcal{L}_a$ using $\widetilde{O}(|\mathcal{Q}_{max}|^2)$ time and $\widetilde{O}(|\mathcal{Q}_{max}|)$ space.
    \end{lemma}

    \begin{proof}
        Schroeppel and Shamir~\cite{DBLP:journals/siamcomp/SchroeppelS81} 
        use priority queues to~implement a~data structure~$D_1$ for the following task:
        output all elements from $\{w(Q_1) + w(Q_2) \colon Q_1 \in \mathcal{Q}_1, Q_2 \in \mathcal{Q}_2\}$ in non-decreasing order using $\widetilde{O}(|\mathcal{Q}_1| \cdot |\mathcal{Q}_2|)$ time and $\widetilde{O}(|\mathcal{Q}_1| + |\mathcal{Q}_2|)$ space. Similarly, a data structure $D_2$ outputs all elements from $\{(w(Q_3) + w(Q_4), Q_4 \cap M) \colon  Q_3 \in \mathcal{Q}_3, Q_4 \in \mathcal{Q}_4\}$ in non-increasing order (we compare elements by their first coordinate) using $\widetilde{O}(|\mathcal{Q}_3| \cdot |\mathcal{Q}_4|)$ time and $\widetilde{O}(|\mathcal{Q}_3| + |\mathcal{Q}_4|)$ space. 
        By $\operatorname{inc}(w(\mathcal Q_1, \mathcal Q_2))$
        and $\operatorname{dec}(w(\mathcal Q_3, \mathcal Q_4))$ we~denote initilization methods for these data structures, by $\operatorname{pop}()$ we~denote their method that 
        gives the next element.
        
        Algorithm~\ref{lst:faster_shamir} solves the problem using data structures
        $D_1$~and~$D_2$.
        It works in time $\widetilde{O}(|\mathcal{Q}_1| \cdot |\mathcal{Q}_2| + |\mathcal{Q}_3| \cdot |\mathcal{Q}_4|)$, uses $\widetilde{O}(|\mathcal{Q}_1| + |\mathcal{Q}_2| + |\mathcal{Q}_3| + |\mathcal{Q}_4|)$ space, and outputs $|\mathcal{L}_a|$ elements. Since the number of outputted subsets is at most $|\mathcal{Q}_4|$, even if we store the output, we still need only $\widetilde{O}(|\mathcal{Q}_{max}|^2)$ time and $\widetilde{O}(|\mathcal{Q}_{max}|)$ space.
        
        \begin{algorithm}
	\DontPrintSemicolon
	\SetKwInOut{Input}{Input}
	\SetKwInOut{Output}{Output}
        \Input{$\mathcal Q_1, \mathcal Q_2, \mathcal Q_3, \mathcal Q_4, a$}
	\Output{$\mathcal L_a = \{(a, Q_4 \cap M) \colon (Q_1, Q_2, Q_3, Q_4) \in \mathcal X_a \}$}
        Initialize $D_1$ = inc($w(\mathcal Q_1, \mathcal Q_2)$) \\
        Initialize $D_2$ = dec($w(\mathcal Q_3, \mathcal Q_4)$) \\
        \While{$(d_2, A) = D_2$\textnormal{.pop()}}{
            \While{$d_1 = D_1$\textnormal{.pop() and} $d_1 + d_2 < a$}{
                \textbf{skip} $d_1$
            }
            \If{$d_1 + d_2 = a$ \textnormal{and} not \textnormal{used[}$A$\textnormal{]}}{
                \textbf{output(}$(a, A)$\textbf{)} \\
                used[$A$] = True
            }
        }
	\caption{Pseudocode for Lemma~\ref{lemma:faster_shamir}.}
	\label{lst:faster_shamir}
        \end{algorithm}
    \end{proof}

    { Note that the size of $\mathcal{L}_a$ is at most $\binom{|M|}{\mu |M|}$.}
    Summarizing, we get the following.

    \begin{lemma}
    \label{lemma:modified_4sum}
        There is an algorithm that given $\mathcal{Q}_1, \mathcal{Q}_2, \mathcal{Q}_3, \mathcal{Q}_4$ and $(r, p, s, q)$, uses $\widetilde{O}\left(|\mathcal{Q}_{max}|^2 + \frac{|\mathcal{Q}_{max}|^4}{pq} + \binom{|M|}{\mu |M|}\right)$ time and $\widetilde{O}\left(|\mathcal{Q}_{max}| + \frac{|\mathcal{Q}_{max}|^4}{pq} + \binom{|M|}{\mu |M|}\right)$ space. It returns a set $\mathcal L \subseteq \mathbb Z \times \mathcal Q_4$ of size $\widetilde O\left(\frac{|\mathcal Q_{max}|^4}{pq} + \binom{|M|}{\mu |M|}\right)$ that satisfies the following. \\
            If there exist a solution $S \subseteq I$ with $w(S) = t$, and a partition $S = S_L \sqcup S_R$, such that:
            \begin{itemize}
                \item $|S| = \alpha n$;
                \item $|M_L \cap S| = |M \cap S| = |M_R \cap S| = \beta n / 2$;
                \item $S \cap M$ is a ($\leq \varepsilon$)-mixer;
                \item $|S \cap L_i| = \gamma |L_i|$ for each $i \in [4]$, the same for $R_i$;
                \item $S_L \subseteq M_L \sqcup L \sqcup M$, $S_R \subseteq M \sqcup R \sqcup M_R$;
                \item $|S_L \cap M| = \mu \beta n$, $|S_R \cap M| = (1/2 - \mu) \beta n$;
                \item $w(S_L) \equiv_p r$, $w(S_L) \equiv_q s$, $w(S_R) \equiv_p t - r$, $w(S_R) \equiv_q t - s$;
                \item $|\mathcal X_{r, p, s, q} \setminus \mathcal X_{w(S_L)}| \leq \ell$,
            \end{itemize}
            then $\mathcal L_{w(S_L)} \subseteq \mathcal L$.
            Otherwise, $\mathcal L$ may be any set.
    \end{lemma}

    Similarly, using Lemma \ref{lemma:modified_4sum}, we find subset $\mathcal{R}$ for the right side. 
    We call WOV($\mathcal{L}, \mathcal{R}$) and return True if it finds a solution. The WOV algorithm works in time $\widetilde O((|\mathcal{L}| + |\mathcal{R}|)2^{\beta n(1 - h(1/4))})$ and space $\widetilde O(|\mathcal{L}| + |\mathcal{R}| + 2^{\beta n})$ by Lemma \ref{lemma:WOV}.

    It is worth noting that for most choices of $s$ our assumption that $s \equiv_{q} w(S_L)$ does not hold. Also, with probability $1 - \Omega^*(1)$, $|\mathcal{X}_{r,p, w(S_L),q} \setminus \mathcal{X}_{w(S_L)}|$ may happen to be greater than~$\ell$. If any of these cases occurs, we may find a~set $\mathcal{L}$ that was not intended (i.e. it does not contain $\mathcal{L}_{w(S_L)}$), or we may even not find anything at all (if after applying the algorithm from Lemma \ref{lemma:shamir_for_rings} we obtain $2\ell+1$ elements, each of which occurs less than $\ell+1$ times). In the second scenario, we simply continue with the next choice of $s$. Otherwise, we cannot determine whether we have found the correct $\mathcal{L}$ or not. However, if an incorrect $\mathcal{L}$ is found, the outcome of our algorithm remains unaffected. If OV($\mathcal{L}, \mathcal{R}$) finds a solution, it is a solution to the original problem; if not, we simply proceed to the next iteration of the loop in line \ref{line:iterating_s2}.

    Now we are ready to present the final Algorithm \ref{lst:main_algorithm}. 
    We will choose the exact values for its parameters later.

\begin{algorithm}
	\DontPrintSemicolon
	\SetKwInOut{Input}{Input}
	\SetKwInOut{Output}{Output}
        \Input{$(I, t)$.}
	\Output{True, if $S \subseteq I$ with $w(S)=t$ exists, otherwise False.}

        \SetKwFor{RepTimes}{repeat}{times}{end}
        \RepTimes{$2^{\lambda n}$}
        {\label{line:repeating_whole_algorithm}
            Select random disjoint subsets $M_L, M, M_R \subseteq I$ of~size $\beta n$
            \\
            \ForEach {$\mu$ \textnormal{such that} $\mu\beta n \in[0, \beta n/2]$}{
            \label{line:foreach_mu}
                Randomly partition $I \setminus (M_L \sqcup M \sqcup M_R)$ into $L, R$ with size as in~\eqref{eq:L_size}--\eqref{eq:R_size} \\
                Pick a~prime $p = \Theta(2^{(1 - \varepsilon) \beta n/2})$ and $r \in \mathbb{Z}_p$ at random \\
                Construct $\mathcal Q_1, \mathcal Q_2, \mathcal Q_3, \mathcal Q_4$ and $\mathcal Q_1', \mathcal     Q_2', \mathcal Q_3', \mathcal Q_4'$
            as defined in~\eqref{eq:q_1_definition}--\eqref{eq:q_4'_definition} \label{line:constructing_qs}\\
                Pick a~random prime 
                $q = \Theta^*\left({2^{n/2}}/(2^{\lambda n}|\mathcal{Q}_{max}|^2)\right)$\\
                
                \ForEach{$s \in \mathbb{Z}_{q}$} { \label{line:iterating_s2}
                    Construct $\mathcal{L}$ and $\mathcal{R}$ (as~in Lemma~\ref{lemma:modified_4sum}) \label{line:constucting_mathcal_l}\\
                    \If{WOV($\mathcal{L}, \mathcal{R}$) (as~in~Lemma~\ref{lemma:WOV})} { \label{line:solving_WOV}
                        \textbf{return} True
                    }
                }
            }
        }
        \textbf{return} False
	\caption{The main algorithm.}
	\label{lst:main_algorithm}
\end{algorithm}

\subsection{Correctness}

\begin{lemma} \label{lemma:probabiliy_of_main_algorithm}
    The success probability of~the algorithm is~$\Omega^*(1)$.
\end{lemma}

\begin{proof}
    To~detect that $I$ is~indeed a~yes-instance, { an} algorithm { iteration} relies on~a~number of~events. 
    The first one is $|(M_L \cup M \cup M_R) \cap S| = 3\beta n/2$ ,that occurs with probability $\Theta^*\left(2^{-\lambda n}\right)$. All the others occur
    with probability $\Omega^*(1)$, conditioned on~the event
    that all the previous events occur. This implies that all of~them occur with probability $\Omega^*\left(2^{-\lambda n}\right)$, and after repeating it $2^{\lambda n}$ times, the success probability of the resulting algorithm is $\Omega^*(1)$. Below, we recall all the considered events and provide links to~lemmas
    proving their conditional probabilities.

    \begin{enumerate}
        \item { $|(M_L \cup M \cup M_R) \cap S| = 3\beta n/2$, see~Eq.~\eqref{eq:lambda}.} \label{ln:correctness_probability_1}
        
        \item $|M_L \cap S| = |M \cap S| = |M_R \cap S| = \frac \beta 2 n$, 
        see Lemma~\ref{lemma:same_fractions}.
        
        \item $|S \cap L_i| = \gamma |L_i|$ and $|S \cap R_i| = \gamma |R_i|$, see Lemma~\ref{lemma:same_fractions}.
        
        \item $p$ is a prime number, see Section \ref{section:primes}. \label{ln:correctness_probability_4}
        \item $\left|\left\{ a \in \mathbb{Z}_p \colon \exists M' \subseteq S \cap M, |M'| = 2\mu |S \cap M|, w(M') \equiv_p a \right\} \right| \ge \Omega\left(p/n^2 \right)$, see Lemma \ref{lemma:p1}.
        \item $|\mathcal X_{r, p}| = O^*(|\mathcal{Q}_{max}|^4 / p)$ and there exists a partition $M_1 \sqcup M_2 = S \cap M$ such that $w(S \cap (M_L \cup L \cup M_1)) \equiv_p r$ and $w(S \cap (M_2 \cup R \cup M_R)) \equiv_p t - r$, see Lemma \ref{lemma:divide_by_p1}.\label{ln:correctness_probability_6}
        \item { $q$ is a prime number, see Section \ref{section:primes}}. \label{ln:correctness_probability_7}
        \item $|\mathcal{X}_{r,p, w(S_L),q} \setminus \mathcal{X}_{w(S_L)}| \le \ell$, see Lemma \ref{lemma:q}.
    \end{enumerate}
    In case any of those events does not happen, the algorithm completes the iteration without affecting the final result.
\end{proof}

\subsection{Setting the Parameters}

Recall that $\alpha = \frac{|S|}{n}$ and we iterate over all possible $\alpha$.
We define $\beta = \beta(\alpha)$ as follows:
\begin{align} \label{eq:beta_definition}
    \beta(\alpha) = \begin{cases}
        0.13, &\alpha \in [0.45, 0.55] \\
        0,    &\text{otherwise.}
    \end{cases}
\end{align}

Let us consider an iteration of the loop on line \ref{line:foreach_mu}. To this point, we know the values of $\varepsilon$, $\varepsilon_L$, $\varepsilon_R$ and $\mu$, so we consider them as fixed numbers.
Since $|M_L| = |M| = |M_R| = \beta n$, we have that $|I \setminus (M_L \sqcup M \sqcup M_R)| = (1 - 3\beta) n$.
Let us consider a~partition $I \setminus (M_L \sqcup M \sqcup M_R) = L_1 \sqcup L_2 \sqcup L_3 \sqcup L_4 \sqcup R_1 \sqcup R_2 \sqcup R_3 \sqcup R_4$ and estimate the size of $\mathcal{Q}_i$ and $\mathcal{Q}'_i$. Recall that $\gamma = \frac{|(L \cup R) \cap S|} {|L \cup R|} = \frac {\alpha - \frac 3 2 \beta} {1 - 3 \beta}$. 
Observe that we~use $\varepsilon$ instead of~$\varepsilon_L$ and $\varepsilon_R$ since $\varepsilon \le \varepsilon_L, \varepsilon_R$.
\begin{align*}
	|\mathcal Q_1| &\le \binom{|L_1|}{\gamma |L_1|} \cdot 2^{(1 - \varepsilon) \beta n} & |\mathcal Q_1'| &\le \binom{|R_1|}{\gamma |R_1|} \cdot 2^{(1 - \varepsilon) \beta n}  \\
	|\mathcal Q_2| &= \binom {|L_2|} {\gamma |L_2|} & |\mathcal Q_2'| &= \binom {|R_2|} {\gamma |R_2|}\\
	|\mathcal Q_3| &= \binom {|L_3|} {\gamma |L_3|} & |\mathcal Q_3'| &= \binom {|R_3|} {\gamma |R_3|}\\
	|\mathcal Q_4| &= \binom{|L_4|}{\gamma |L_4|} \cdot \binom{|M|}{\mu \beta n} & |\mathcal Q_4'| &= \binom{|R_4|}{\gamma |R_4|} \cdot \binom{|M|}{\left(\frac 1 2 - \mu\right) \beta n}.
\end{align*}

We construct $\mathcal{L}$ and $\mathcal{R}$ in time depending on $|\mathcal{Q}_{max}|^2 + |\mathcal{Q}'_{max}|^2$. Note that this function is minimized when all $|\mathcal{Q}_i|$ and $|\mathcal{Q}'_i|$ are roughly the same. We can achieve that by a careful choice of $|L_i|$ and $|R_i|$.

We define the sizes of $|L|$ and $|R|$ as shown below and require that $|L_1| + |L_2| + |L_3| + |L_4| = |L|$ and $|R_1| + |R_2| + |R_3| + |R_4| = |R|$. \begin{align}
   |L| &= \frac {1 - 3 \beta - \chi \beta}{2} n \label{eq:L_size}, \\
   |R| &= \frac {1 - 3 \beta + \chi \beta}{2} n \label{eq:R_size},
\end{align}

where $\chi = \frac {h(\mu) - h\left(\frac 1 2 - \mu\right)} {h(\gamma)}$ is a~balancing parameter 
needed to~ensure that 
$|\mathcal{Q}_{max}| \approx |\mathcal{Q}'_{max}|$.

Now we will show that all the parameters are chosen correctly and will define sizes of $L_1, L_2, L_3, L_4, R_1, R_2, R_3, R_4$, so we never, for example, choose the size for $L_i$ to be negative, or repeat the algorithm a number of times smaller than one.

 \begin{lemma}
     $\lambda$ defined in Equation \eqref{eq:lambda} is a~non-negative number.
 \end{lemma}

 \begin{proof}
     Entropy is a~concave function, hence
 \[
     \alpha h\left(\frac{3\beta}{2\alpha}\right) + (1 - \alpha) h\left(\frac{3\beta}{2(1-\alpha)}\right) \le h\left(\frac{3\beta}{2} + \frac{3\beta}{2}\right) = h(3\beta) \; .
 \]

 This implies that $\lambda \ge 0$.
 \end{proof}

 \begin{lemma} \label{lemma:L_size_non_negative}
     For our choice of $\beta$, for every $\alpha \in [0, 1]$, $|L| \ge 0$, where
     \[
         |L| = n \frac{h(\gamma) - 3 \beta h(\gamma) - \beta h(\mu) + \beta h(1/2 - \mu)}{2 h(\gamma)}\tag{Eq.~\eqref{eq:L_size}} \; .
     \]
 \end{lemma}

 \begin{proof}
     \begin{multline*}
         |L| \ge 0 \iff h(\gamma) - 3 \beta h(\gamma) + \beta (-h(\mu) + h(1/2-\mu)) \ge 0 \\
         h(\gamma) - 3 \beta h(\gamma) + \beta (\underbrace{-h(\mu) + h(1/2-\mu)}_{\ge -1}) \ge h(\gamma) - 3 \beta h(\gamma) - \beta = h(\gamma)(1 - 3 \beta) - \beta
     \end{multline*}

     Consider two cases:

     \begin{enumerate}
         \item $\alpha \in [0.45, 0.55]$, then $\beta = 0.13, h(\gamma) \ge h(0.45) \ge 0.99$ and
 \[
     h(\gamma)(1 - 3 \beta) - \beta \ge 0.99 (1 - 3 \cdot 0.13) - 0.13 \ge 0.47
 \]

         \item $\alpha \in [0, 0.45) \cup (0.55, 1]$, then $\beta = 0$ and $|L| = \frac 1 2 n$.
     \end{enumerate}
 \end{proof}

 \begin{lemma}
 	For our choice of $\beta$, for every $\alpha \in [0,1]$ we can choose $|L_1|$, $|L_2|$, $|L_3|$, $|L_4|$, $|R_1|$, $|R_2|$, $|R_3|$, $|R_4|$ $\geq 0$ such that $|L_1| + |L_2| + |L_3| + |L_4| = |L|$, $|R_1| + |R_2| + |R_3| + |R_4| = |R|$, and
 \begin{align}
 	\log |\mathcal Q_{max}| &\leq n \cdot \max \{\frac{h(\gamma) - 3 \beta h(\gamma) + 2 \beta - 2 \varepsilon \beta + \beta h(\mu) + \beta h\left(1/2 - \mu\right)} {8}, \; (1 - \varepsilon) \beta\} \; . \label{eq:q_max}
 \end{align}
 \end{lemma}

 \begin{proof}
 	To make all $|\mathcal Q_i|$ and $|\mathcal Q_i'|$ equal $n \frac{(h(\gamma) - 3 \beta h(\gamma) + 2 \beta - 2 \varepsilon \beta + \beta h(\mu) + \beta h\left(1/2 - \mu\right))} {8}$ we set sizes of  $L_i$ and $R_i$ as follows:
 \begin{align}
     |L_1| &= n \frac {h(\gamma) - 3 \beta h(\gamma) - 6 \beta + 6 \varepsilon \beta + \beta h(\mu) + \beta h\left(\frac 1 2 - \mu\right)} {8 h(\gamma)} = |R_1| \label{eq:l1_size}\\
     |L_2| = |L_3| &= n \frac {h(\gamma) - 3 \beta h(\gamma) + 2 \beta - 2 \varepsilon \beta + \beta h\left(\mu\right) + \beta h\left(\frac 1 2 - \mu\right)} {8 h(\gamma)} = |R_2| = |R_3| \label{eq:l2_size}\\
     |L_4| &= n\frac{h(\gamma) - 3 \beta h(\gamma) + 2 \beta - 2 \varepsilon \beta - 7 \beta h(\mu) + \beta h\left(\frac 1 2 - \mu \right)}{8 h(\gamma)} \label{eq:l4_size}\\
     |R_4| &= n\frac{h(\gamma) - 3 \beta h(\gamma) + 2 \beta - 2 \varepsilon \beta + \beta h(\mu) - 7 \beta h\left(\frac 1 2 - \mu \right)}{8 h(\gamma)} \label{eq:r4_size}
 \end{align}

 If all $|L_i|$ and $|R_i|$ are non-negative, then the statement of the lemma holds.
 Assume that some values are negative.
 In this case we replace all $|L_i| < 0$ with $|L_i| = 0$ and scale all the $|L_i|$ in such a way that $|L_1| + |L_2| + |L_3| + |L_4| = |L|$. Then we do the same procedure with $|R_i|$. 
 We can see that for every $|L_i| = 0$, $|\mathcal{Q}_i| \le 2^{(1 - \varepsilon) \beta n}$, and for $|L_i| > 0$, $|\mathcal{Q}_i| \le n \frac{(h(\gamma) - 3 \beta h(\gamma) + 2 \beta - 2 \varepsilon \beta + \beta h(\mu) + \beta h\left(1/2 - \mu\right))} {8}$. The same works for $|\mathcal Q'_i|$. Thus, the lemma holds.
 \end{proof}

We have set all the parameters and now we are ready to estimate the time and space complexity of Algorithm \ref{lst:main_algorithm}.

\subsection{Time Complexity}

The running time is~dominated~by the following parts.

\begin{itemize}
    \item Repeating the whole algorithm $2^{\lambda n}$ times (Line \ref{line:repeating_whole_algorithm}).
    \item Constructing $\mathcal Q_1, \mathcal Q_2, \mathcal Q_3, \mathcal Q_4$ and $\mathcal Q_1', \mathcal Q_2', \mathcal Q_3', \mathcal Q_4'$ takes $O^*(|\mathcal Q_{max}| + 2^{\beta n})$ time (Line \ref{line:constructing_qs}).
    \item Repeating $q$ times (Line \ref{line:iterating_s2}).
    \item Finding $\mathcal L, \mathcal R$ using Lemma~\ref{lemma:modified_4sum} in time $O^*\left(|\mathcal Q_{max}|^2 + \frac {|\mathcal Q_{max}|^4} {p q} + \binom{|M|}{\max\{\mu, 1/2-\mu\} |M|}\right)$ 
    (Line~\ref{line:constucting_mathcal_l}).
    \item Solving WOV$(\mathcal L, \mathcal R)$ using 
    Lemma~\ref{lemma:WOV} in time $O^*\left((|\mathcal L| + |\mathcal R|) 2^{\beta n(1 - h(1/4))}\right)$ (Line \ref{line:solving_WOV}). The $O^*$ here hides logarithmic factor in $|\mathcal L| + |\mathcal R|$ that grows polynomially in~$n$.
\end{itemize} 

Let us consider an iteration of the loop on line \ref{line:foreach_mu}. Note that the values of $p$, $q$, $|\mathcal{Q}_{max}|$ and $\mu$ are now determined. The running time of Algorithm \ref{lst:main_algorithm} on this specific iteration is
\[
    O^*\left(q \cdot \left(|\mathcal{Q}_{max}|^2 + \frac{|\mathcal{Q}_{max}|^4}{p q} + \left(\frac{|\mathcal{Q}_{max}|^4}{p q} + 2^{\beta n h(\max\{\mu, 1/2-\mu\})}\right)2^{\beta n(1 - h(1/4))}\right)\right) \, .
\] 

To bound the total running time of Algorithm \ref{lst:main_algorithm}, we can multiply the number of the iterations $\Theta^*(2^{\lambda n})$ by the upper bound on the slowest iteration.

 \begin{lemma} \label{lemma:time_complexity_main_algo}
 	The total running time of the Algorithm \ref{lst:main_algorithm} is $O^*(2^{n/2} + 2^{T(\alpha, \beta) \cdot n})$, where we define $T(\alpha, \beta)$ as follows:
 \begin{equation}
	 T(\alpha, \beta) = \frac{1}{2} \left(h(\gamma) - 3 \beta h(\gamma) + 3 \beta + 2 \lambda \right) . \label{eq:time_complexity}
 \end{equation}
 Here, $\gamma = \gamma(\alpha)$ and $\lambda=\lambda(\alpha)$. 
 \end{lemma}

 \begin{proof}
  As it was shown the total running time does not exceed
 \[
     O^*\left(2^{\lambda n} \cdot q \left(|\mathcal{Q}_{max}|^2 + \frac{|\mathcal{Q}_{max}|^4}{p q} + (\frac{|\mathcal{Q}_{max}|^4}{p q} + 2^{\beta n h(\max\{\mu, 1/2-\mu\})})2^{\beta n(1 - h(1/4))}\right)\right) \; ,
 \]
 where the values of $p$, $q$, $|\mathcal Q_{max}|$ and $\mu$ correspond to the slowest iteration of the loop on line \ref{line:foreach_mu}.

 Since $p$ and $q$ can vary only with polynomial factors, here by them we will denote the upper and lower bounds on their values, so let $\log p = \frac{(1 - \varepsilon)}{2}\beta n$ and 
 $\log q = \frac{1}{2}n - \lambda n - 2\log |\mathcal{Q}_{max}|$.

 Let us estimate the logarithm of the first summand inside the parenthesis:
 \[
     \log (2^{\lambda n} \cdot q \cdot |\mathcal Q_{max}|^2) = \lambda n + \underbrace{\frac 1 2n - \lambda n - 2 \log |\mathcal Q_{max}|}_{\log q}  + 2 \log |\mathcal Q_{max}| = \frac 1 2 n.
 \]

 We~estimate the last term of the sum (the fourth one) assuming that $\mu \ge \frac 1 2 - \mu$, the other case can be shown similarly: 
 \[
     \log\left(2^{\beta n h(\mu) + \beta n(1 - h(1/4))}\right) = \beta n h(\mu) + \beta n - \beta n h(1/4) \le 2 \beta n \; .
 \] 

 Recall that we choose $\beta < 0.15$, then $2 \beta n \le 0.5 n$.

 The second summand is at~most the third one, so it~suffices to~estimate the logarithm of the third summand. Recall that that $\log |\mathcal Q_{max}| \leq n \cdot \max \{\frac{h(\gamma) - 3 \beta h(\gamma) + 2 \beta - 2 \varepsilon \beta + \beta h(\mu) + \beta h\left(1/2 - \mu\right)} {8}, \; (1 - \varepsilon) \beta\} $. If $\log |\mathcal Q_{max}| \leq (1 - \varepsilon) \beta n$, then the third summand is dominated by $\frac 1 2 (9 \beta - 2 \beta h(\frac 1 4))$ and substituting $\beta = 0.13$ yields value $\approx 0.48 < 0.5$. Therefore, this case is not considered further.
 \begin{align*}
     &\phantom{=}\log (2^{\lambda n} q \frac {|\mathcal Q_{max}|^4}{p q} 2^{\beta n(1 - h(1/4))})\\
     &=  \lambda n + 4 \log |\mathcal Q_{max}| - \underbrace{\frac{1 - \varepsilon} 2 \beta n}_{\log p} + \beta n - \beta n h(1/4)\\
     &= \lambda n + n \frac{h(\gamma) - 3 \beta h(\gamma) + 3 \beta -\varepsilon \beta + \beta h(\mu) + \beta h(1/2 - \mu) - 2 \beta h(1/4)}{2} \; .
 \end{align*}

 Having $(- \varepsilon \beta) \le 0$, $h(\mu) + h(1/2 - \mu) \le 2 h(1/4)$, this expression can be dominated by
 \[
 \lambda n + n \frac {h(\gamma) - 3 \beta h(\gamma) + 3 \beta} {2}.
 \]
 So, the total running time of the algorithm is $O^*(2^{n/2} + 2^{T(\alpha, \beta) \cdot n})$, where
 \begin{equation}
     T(\alpha, \beta) = \frac{1}{2} \left(h(\gamma) - 3 \beta h(\gamma) + 3 \beta + 2 \lambda \right) .
 \end{equation}

 \end{proof}

\begin{corollary} \label{corollary:T_alpha_beta}
    If for every $\alpha \in [0, 1]$, $T(\alpha, \beta(\alpha)) \le 0.5$, then Algorithm \ref{lst:main_algorithm} has running time $O^*\left(2^{0.5 n}\right)$.
\end{corollary}

We complete the analysis in Section \ref{sec:parameter_substitution}.

\subsection{Space Complexity}

The space usage is~dominated~by the following parts.

\begin{itemize}
    \item Constructing $\mathcal Q_1, \mathcal Q_2, \mathcal Q_3, \mathcal Q_4$ and $\mathcal Q_1', \mathcal Q_2', \mathcal Q_3', \mathcal Q_4'$ with $O^*(|\mathcal Q_{max}|)$ space (Line \ref{line:constructing_qs}).
    \item Finding $\mathcal L, \mathcal R$ using Lemma \ref{lemma:modified_4sum} with $O^*\left(|\mathcal Q_{max}| + \frac {|\mathcal Q_{max}|^4} {p q} + \binom{|M|}{\max\{\mu, 1/2-\mu\} |M|}\right)$ space (Line \ref{line:constucting_mathcal_l}).
    \item Solving WOV$(\mathcal L, \mathcal R)$ using Lemma \ref{lemma:WOV} with $O^*\left(|\mathcal L| + |\mathcal R| + 2^{\beta n}\right)$ space (Line \ref{line:solving_WOV}).
\end{itemize} 

For a specific iteration of the loop on line \ref{line:foreach_mu}, the space usage is at most the sum of the above expressions, where $|\mathcal L|, |\mathcal R| = O^*(\frac {|\mathcal Q_{max}|^4} {p q} + \binom{|M|}{\max\{\mu,1/2-\mu\}}|M|)$. The total space complexity of Algorithm \ref{lst:main_algorithm} can be bounded by the space usage of an iteration, requiring the most space.

Therefore, it suffices to~prove that $2^{\beta n}$, $\binom{|M|}{\max\{\mu,1/2-\mu\}|M|}$, $|\mathcal Q_{max}|$ and $\frac {|\mathcal Q_{max}|^4} {p q}$ are at most $O^*(2^{0.246n})$ on each iteration.

We can see that the following lemma holds.
\begin{lemma} \label{lemma:S_alpha_beta}
    If for every $\alpha \in [0, 1]$, $S(\alpha, \beta(\alpha)) \le 0.246$, then the total space usage of Algorithm \ref{lst:main_algorithm} does not exceed $O^*\left(2^{0.246 n}\right)$.
\end{lemma}

 \begin{lemma}\label{lemma:space_main_algorithm}
 	The total space usage of Algorithm \ref{lst:main_algorithm} is $O^*(2^{0.246n} + 2^{S(\alpha, \beta) \cdot n})$, where $S(\alpha, \beta)$ is defined as follows:
 \begin{align}
	 S(\alpha, \beta) = \frac 1 4 \left( 3 h(\gamma) - 9 \beta h(\gamma) + 6 \beta h(\frac 1 4) + 4 \beta - 2 + 4\lambda \right). \label{eq:space_complexity}
 \end{align}
 Here, we interpret $\gamma$ as $\gamma = \gamma(\alpha)$ and $\lambda$ as $\lambda(\alpha)$.
 \end{lemma}

 \begin{proof}
 As we have already discussed, it is sufficient to estimate the following values: $2^{\beta n}$, $\binom{|M|}{\max\{\mu,1/2-\mu\}|M|}$, $|\mathcal Q_{max}|$ and $\frac {|\mathcal Q_{max}|^4} {p q}$, 
 { for the iteration of the loop in line \ref{line:foreach_mu} that requires the most space.}
 To do that, we estimate their logarithms.
 The first one is way smaller than $0.246 n$, since $\beta < 0.15$.
 The second one is equal to $\beta n h(\mu) \le \beta n \le 0.246 n$.

 Recall that that $\log |\mathcal Q_{max}| \leq n \cdot \max \{\frac{h(\gamma) - 3 \beta h(\gamma) + 2 \beta - 2 \varepsilon \beta + \beta h(\mu) + \beta h\left(1/2 - \mu\right)} {8}, \; (1 - \varepsilon) \beta\} $. If $\log |\mathcal Q_{max}| \leq (1 - \varepsilon) \beta n$, then the third value is dominated by $\frac 1 2 (11 \beta - 1)$ and substituting $\beta = 0.13$ yields value $\approx 0.215 < 0.246$. Therefore, this case is not considered further.

 \[
 \log |\mathcal Q_{max}| =
 \frac{1}{8} \left(h(\gamma) - 3 \beta h(\gamma) + 2 \beta - 2 \varepsilon \beta + \beta h\left(\mu\right) + \beta h(\frac 1 2 - \mu)\right) n \; .
 \]
 Having $(-2 \varepsilon \beta) \le 0$, $h(\mu) + h(1/2 - \mu) \le 2$ and $\beta < 0.15$, we get
 \[
 \log |\mathcal Q_{max}| \le \frac 1 8 (1 + \beta) n < 0.246 n.
 \]
 It is left to estimate $\log \frac {|\mathcal Q_{max}|^4} {p q}$.
 \begin{align*}
     &\log \frac {|\mathcal Q_{max}|^4} {p q} = 4 \log |\mathcal Q_{max}| -\underbrace{\frac {1 - \varepsilon} 2 \beta n}_{p} - \underbrace{\frac 1 2 n + 2 \log |\mathcal Q_{max}| + \lambda n}_{q} = \\
     &= \lambda n + n \frac{3 h(\gamma) - 9 \beta h(\gamma) + 3 \beta h(\mu) + 3 \beta h(1/2 - \mu) + 4 \beta - 4 \varepsilon \beta - 2}{4}
 \end{align*}
 Having $(-4 \varepsilon \beta) \le 0$ and $h(\mu) + h(1/2 - \mu) \le 2 h(1/4)$, we dominate this expression by
 \[\lambda n + \frac {3 h(\gamma) - 9 \beta h(\gamma) + 6 \beta h(1/4) + 4 \beta - 2} {4} \; .\]

 That concludes the proof.
 \end{proof}

\begin{corollary} \label{corollary:S_alpha_beta}
    If for every $\alpha \in [0, 1]$, $S(\alpha, \beta(\alpha)) \le 0.246$, then the total space usage of Algorithm \ref{lst:main_algorithm} does not exceed $O^*\left(2^{0.246 n}\right)$.
\end{corollary}

We complete the analysis in the following section.

\subsection{Parameter Substitution} \label{sec:parameter_substitution}

\begin{restatable}{lemma}{maxpointtime}[Maximum point for function of time complexity]\label{lemma:maximum_point_time}
    The function $T(\alpha, \beta)$ (see~\eqref{eq:time_complexity}) satisfies $T(\alpha, \beta) \le T(0.5, 0.13)$ when $\beta$ is defined as~in~\eqref{eq:beta_definition}.
\end{restatable}

 \begin{proof}[Proof of Lemma \ref{lemma:maximum_point_time}]
   Note that
   \begin{align*}
   T(\alpha,\beta) &= \frac{1}{2}\left(h(\gamma)-3\beta h(\gamma)+3\beta + 2(\alpha h\left(\frac{3\beta}{2\alpha}\right) + (1-\alpha) h\left(\frac{3\beta}{2(1-\alpha)}\right) - h(3\beta)))\right),\\
 \frac{\partial^2 T(\alpha,\beta)}{\partial \alpha^2} &= 
     \frac{18\alpha^2 \beta-2\alpha^2-18\alpha \beta-9\beta^2+2\alpha+6\beta}{\alpha(2\alpha-3\beta)(1-\alpha)(-2+3\beta+2\alpha)\ln(2)}.
     \end{align*}

     Consider two cases (the case $\alpha>0.5$ is symmetrical).

     \begin{enumerate}
         \item $0.45 \le \alpha \le 0.5$.

            In this regime, we have $\beta = 0.13$, the roots of the polynomials of the denominator and numerator are either $< 0.2$ or $\geq 0.5$. By checking that $\frac{\partial T(\alpha,\beta)}{\partial \alpha}|_{a=\frac{1}{2}}=0$ and $\frac{\partial^2 T(\alpha,\beta)}{\partial \alpha^2}|_{a=\frac{1}{2}}=0$, we verify that $a=\frac{1}{2}$ is indeed the maximum and have
         \[T(\alpha,0.13) \leq T(0.5,0.13) \, .\]

         \item $\alpha < 0.45$.

         In~this case, $\beta = 0$. Then all the roots of the polynomials of the denominator and numerator are either~$0$ or~$1$.  By checking that $\frac{\partial^2 T(\alpha,\beta)}{\partial^2 \alpha}|_{\alpha \in [0, 0.45]}<0$ and  $\frac{\partial T(\alpha,\beta)}{\partial \alpha}|_{\alpha \in [0, 0.45]} > 0$, we get that
         \[T(\alpha,0)\leq T(0.45,0) \leq T(0.5, 0.13) \, .\]
     \end{enumerate}
 \end{proof}

\begin{restatable}{lemma}{maxpointspace}[Maximum point for function of space complexity]\label{lemma:maximum_point_space}
    The function $S(\alpha, \beta)$ (see~\eqref{eq:space_complexity}) satisfies $S(\alpha, \beta) \le S(0.5, 0.13)$ when $\beta$ is defined as~in~\eqref{eq:beta_definition}.
\end{restatable}

 \begin{proof}[Proof of Lemma \ref{lemma:maximum_point_space}]
 Recall that
 \begin{multline*}
 	S(\alpha, \beta) = \frac{1}{4} (3 h(\gamma) - 9 \beta h(\gamma) + 6 \beta h\left(\frac{1}{4}\right) + 4 \beta - 2 + \\ 4(\alpha h\left(\frac{3\beta}{2 \alpha}\right) + (1 - \alpha)h\left(\frac{3\beta}{2(1 - \alpha)}\right) - h(3 \beta))).
 \end{multline*}

     Consider two cases ($\alpha>0.5$ is symmetrical). 
   
     \begin{enumerate}
         \item $0.45 \leq \alpha \leq 0.5$. Then, $\beta = 0.13$. 

     \[\frac{\partial^2 S(\alpha,\beta)}{\partial \alpha^2} =
 -\frac{3(7\alpha^2 \beta-\alpha^2-7\alpha \beta-3\beta^2+\alpha+2\beta)}{\alpha(\alpha-1)(2\alpha-3\beta)(3\beta+2\alpha-2)\ln(2)}\]

 One can check that all the roots of the numerator and denominator
 are either $<0.2$ or $>0.5$. As $\frac{\partial^2 S(\alpha,\beta)}{\partial \alpha^2}|_{\alpha=\frac{1}{2}} < 0$, we conclude that $S$ is concave on the interval $[0.45,0.5]$. The only thing left to do is to verify that $\frac{\partial S(\alpha,\beta)}{\partial \alpha}|_{\alpha=\frac{1}{2}}=0$ which implies that $S(\alpha,0.13) \leq S(\frac{1}{2},0.13)$ if $0.45 \leq \alpha\leq 0.5$.

         \item $\alpha<0.45$. Then, $\beta = 0$.

         \[S(\alpha,0) = -\frac{1}{4}(3\alpha \log(\alpha)-3\alpha\log(1-\alpha)+2\log(2)+3 \log(1-\alpha))\]

 \[\frac{\partial^2 S(\alpha,0)}{\partial \alpha^2} = -\frac{3}{4\alpha (1-\alpha) \ln(2)} < 0\]

 Checking that $\frac{\partial S(\alpha,0)}{\partial \alpha}|_{\alpha \in [0, 0.45]} > 0$ we get that it is maximal at $\alpha = 0.45$ and $S(\alpha,0) \leq S(0.45,0) \leq S(0.5, 0.13)$. This concludes the proof.
     \end{enumerate}
 \end{proof}

Hence, the time and space complexity do not exceed $T(0.5, 0.13)$ and $S(0.5, 0.13)$, respectively (for any $\alpha$ and $\beta$ as in~\eqref{eq:beta_definition}). 
Recall that $\lambda = h(3\beta) - \alpha h(\frac{3\beta}{2 \alpha}) - (1 - \alpha)h(\frac{3\beta}{2(1 - \alpha)})$, 
and $\gamma = \frac{\alpha - \frac 3 2 \beta}{1 - 3\beta}$, see Equation \eqref{eq:gamma_definition}.
By~plugging $\alpha = 0.5, \beta = 0.13$, we~get the following values for the parameters:
\begin{align*}
    \gamma &= 0.5,\\
    h(\gamma) &= 1,\\
    \lambda &= 0, \\
    T(0.5, 0.13) &= \frac 1 2 \cdot (1 - 3 \cdot 0.13 + 3 \cdot 0.13) = \frac 1 2, &\\
    S(0.5, 0.13) & \le  \frac 1 4 \cdot (3 - 9 \cdot 0.13 + 6 \cdot 0.13 \cdot 0.812 + 4 \cdot 0.13 - 2) \le 0.246 .&
\end{align*}

Combining statements \ref{corollary:T_alpha_beta} and \ref{corollary:S_alpha_beta}, we get that the algorithm runs in time  $O^*\left(2^{0.5 n}\right)$ and space $O^*\left(2^{0.246 n}\right)$, which completes the proof of Theorem \ref{theorem:main_theorem}.

\subsection{Algorithm for Weighted Orthogonal Vectors} \label{sec:wov_lemmas}

 \WOV*
 \begin{definition}[Orthogonal Vectors, OV] \label{definition:ov}
     Given families of vectors $\mathcal A, \mathcal B \subseteq 2^{[d]}$. Decide if there exists a pair $A \in \mathcal A, B \in \mathcal B$ such that $A \cap B = \emptyset$.
 \end{definition}

 \begin{lemma} \label{lemma:OV} 
     For any $\mu \in [0, \frac 1 2]$, there is a Monte-Carlo algorithm that, given $\mathcal A \subseteq \binom{[d]}{\mu d}$ and $\mathcal B \subseteq \binom{[d]}{(1/2 - \mu)d}$, solves $\operatorname{OV}(\mathcal{A}, \mathcal{B})$ in time
 \[
     \widetilde O\left((|\mathcal A| + |\mathcal B|) 2^{d(1-h(1/4))}\right)
 \]
 and space $\widetilde O\left(|\mathcal A| + |\mathcal B| + 2^d\right)$.

 \end{lemma}

 \begin{proof}
     In case $\mu \in [0.2, 0.3]$, this becomes exactly Theorem~6.1 from~\cite{DBLP:conf/stoc/NederlofW21} applied for $\alpha = \frac 1 2, \sigma = 2 \mu$.
     The restriction $\sigma \in [0.4, 0.6]$ in that theorem appears due to the following lemma.

 \begin{lemma}[Lemma B.1 from \cite{DBLP:conf/stoc/NederlofW21}]
     \label{lemma:b1}
     For large enough $n$ and $\lambda \in [0.4, 0.5]$ and $\sigma \in [0.4, 0.6]$, 
     the following inequality holds:
     \[
         \min_x \left\{
              \frac{\binom{1-\lambda\sigma}{x - \lambda\sigma}_n +
              \binom{1-(1-\sigma)\lambda}{x}_n}{\binom{1-\lambda}{x -
              \lambda\sigma}_n}
         \right\} \le
         2^{n(1/2 + \lambda - h(\lambda/2))} n^{O(1)}.
     \] 
     Here, $\binom{\alpha}{\beta}_n$ denotes $\binom{\alpha n}{\beta n}$.
 \end{lemma}

 We need that condition to hold for any $\sigma \in [0, 1]$ when $\lambda = 0.5$.
 For this reason, we adjust the statement of the lemma.

 \begin{lemma}\label{lemma:b1_alternative}
     For large enough $n$ and $\sigma \in [0, 1]$ the following inequality holds:
     \[
         \min_x \left\{
              \frac{\binom{1-\sigma/2}{x - \sigma/2}_n +
              \binom{1-(1-\sigma)/2}{x}_n}{\binom{1/2}{x -
              \sigma/2}_n}
         \right\} \le
         2^{n(1 - h(1/4))} n^{O(1)}.
     \]
 \end{lemma}

 \begin{proof}
     We will use $\mu = \frac{\sigma}{2}$, $\mu \in [0, 0.5]$ to simplify the calculations.
    \[
         \min_x \left\{
              \frac{\binom{1-\mu}{x - \mu}_n +
              \binom{\frac{1}{2}+\mu}{x}_n}{\binom{1/2}{x -
              \mu}_n}
         \right\} = O^*(\min_x \{2^{\max \{ (1-\mu)h(\frac{x-\mu}{1-\mu}) , (\frac{1}{2}+\mu)h(\frac{2x}{1+2\mu})\}-\frac{1}{2}h(2x-2\mu)}\}) \leq
     \]

    (we use $x = \frac{1}{2} + \log(3)(\mu-\frac{1}{4})$ as a good enough approximation of the minimum)
    \[\leq O^*(2^{max \{ (1-\mu)h(\frac{( \frac{1}{2} + \log(3)(\mu-\frac{1}{4}))-\mu}{1-\mu}) , (\frac{1}{2}+\mu)h(\frac{2( \frac{1}{2} + \log(3)(\mu-\frac{1}{4}))}{1+2\mu})\}-\frac{1}{2}h(2( \frac{1}{2} + \log(3)(\mu-\frac{1}{4}))-2\mu)})\]  

     Let:
    \[f_1 = (1-\mu)h\left(\frac{\left( \frac{1}{2} + \log(3)\left(\mu-\frac{1}{4}\right)\right)-\mu}{1-\mu}\right) -\frac{1}{2}h\left(2\left( \frac{1}{2} + \log(3)\left(\mu-\frac{1}{4}\right)\right)-2\mu\right) \; ,\]
 \[f_2 =  \left(\frac{1}{2}+\mu\right)h\left(\frac{2\left( \frac{1}{2} + \log(3)\left(\mu-\frac{1}{4}\right)\right)}{1+2\mu}\right) -\frac{1}{2}h\left(2\left( \frac{1}{2} + \log(3)\left(\mu-\frac{1}{4}\right)\right)-2\mu\right) \; .\]

 We will show that both of these functions have a maximum at $\mu = \frac{1}{4}$. Indeed,
 \[\frac{\partial^2 f_1}{\partial \mu^2} = \frac{(3\log{3}-2)^2}{(1-\mu)(4\log{3}\mu - 4\mu - \log{3}+2)(4\log(3)\mu-\log(3) -2 )} \; ,\]
 \[\frac{\partial^2 f_2}{\partial \mu^2} = -\frac{8(\log{3}-1)^2}{(4\log(3)\mu - 4\mu - \log{3}+2)(4\log(3)\mu-4\mu-\log(3) )} \; .\]

 By analyzing the roots of the polynomials in the denominators one can show that for $0 \leq \mu \leq \frac{1}{2}: \frac{\partial^2 f_1}{\partial \mu^2}<0$ and $ \frac{\partial^2 f_2}{\partial \mu^2}<0$, which means that there is at most one maximum for each of the functions on this interval. Now we just need to notice that
 \[ \frac{\partial f_1}{\partial \mu}|_{\mu=\frac{1}{2}} =\frac{\partial f_2}{\partial \mu}|_{\mu=\frac{1}{2}} = 0\]

 and $f_1 |_{mu=\frac{1}{2}} = f_2 |_{mu=\frac{1}{2}} = 1 - h(\frac{1}{4})$. This concludes the proof.
 \end{proof}

 Now we can apply the proof of Theorem 6.1 from~\cite{DBLP:conf/stoc/NederlofW21} with Lemma \ref{lemma:b1} replaced by Lemma \ref{lemma:b1_alternative}. That completes the proof of Lemma \ref{lemma:OV}.
 \end{proof}

 \lemmawov*
 \begin{proof}[Proof of Lemma \ref{lemma:WOV}] \label{proof:WOV}
     To prove the lemma we show a reduction from WOV to OV.
     We split $\mathcal{A}$ and $\mathcal{B}$ into the following groups of sets with the same sum in time $\widetilde{O}(|\mathcal{A}| + |\mathcal{B}|)$.
 \begin{align*}
     \mathcal L(a) &= \{A \colon (A, a) \in \mathcal A \} \; , \\
     \mathcal R(b) &= \{B \colon (B, b) \in \mathcal B \} \; .
 \end{align*}
     As in the algorithm for $2$-SUM, we use the two pointers technique to iterate over every possible $(a, b)$ such that $a + b = t$ in time $\widetilde{O}(\text{number of groups})$.
     For every such pair we call $\operatorname{OV}(\mathcal L(a), \mathcal R(b))$.
     By Lemma~\ref{lemma:OV}, its running time and space usage depend linearly on $(|\mathcal L(a)| + |\mathcal R(b)|)$, so the total time and space complexity are at most as for $\operatorname{OV}(\mathcal{A}, \mathcal{B})$.
 \end{proof}

\section*{Acknowledgments}
Research is~partially supported
by the grant 075-15-2022-289 for creation and development of Euler International Mathematical Institute,
and 
by the Foundation for the Advancement of Theoretical
Physics and Mathematics ``BASIS''.


\end{document}